\documentclass[fleqn,10pt]{wlscirep}
\usepackage[utf8]{inputenc}
\usepackage{amsthm}
\usepackage[T1]{fontenc}
\usepackage{bm}
\usepackage{mathtools}

\setlist[enumerate]{leftmargin=*}
\setlist[itemize]{leftmargin=*}

\newtheorem{theorem}{Theorem}

\newtheorem{remark}{Remark}

\title{Variable-moment fluid closures with Hamiltonian structure}

\author[1,*]{J. W. Burby}
\affil[1]{Los Alamos National Laboratory, Theoretical Division, Los Alamos, NM 87545, USA}
\affil[*]{jburby@lanl.gov}

\begin{abstract}
Based on ideas due to Scovel-Weinstein, I present a general framework for constructing fluid moment closures of the Vlasov-Poisson system that exactly preserve that system's Hamiltonian structure. Notably, the technique applies in any space dimension and produces closures involving arbitrarily-large finite collections of moments. After selecting a desired collection of moments, the Poisson bracket for the closure is uniquely determined. Therefore data-driven fluid closures can be constructed by adjusting the closure Hamiltonian for compatibility with kinetic simulations.
\end{abstract}
\begin{document}

\flushbottom
\maketitle
\thispagestyle{empty}

%\noindent Please note: Abbreviations should be introduced at the first mention in the main text – no abbreviations lists. Suggested structure of main text (not enforced) is provided below.

\section{Introduction}
The simplest non-dissipative kinetic model for a single-component plasma in $N$ space dimensions is the Vlasov-Poisson system,
\begin{gather*}
\partial_tf + m^{-1}p_i\,\partial_{q^i} f - e\,m^{-1}\partial_{q^i} V\partial_{p_i}f =0,\quad f:\mathbb{R}_q^N\times \mathbb{R}_p^N\rightarrow\mathbb{R} \\
-\Delta V = 4\pi e \int f\,dp - 4\pi e n_0  ,\quad V:\mathbb{R}_q^N\rightarrow \mathbb{R}.
\end{gather*}
Here $f=f(q,p)$ denotes the Vlasov distribution function and $V = V(q)$ is the electrostatic potential. Particles that comprise the sole plasma species have charge $e$, mass $m$, position $q = (q^1,\dots,q^N)$, and momentum $p=(p_1,\dots,p_N)$. A uniform neutralizing background $n_0$ ensures overall charge neutrality. The explicit time dependence of $f$ and $V$ has been suppressed. This single-component model generalizes simply to allow for multiple plasma species.

Kinetic plasma models, like the Vlasov-Poisson system, are formally equivalent to infinite hierarchies of hydrodynamic equations for the fluid moments. The degree-$k$ fluid moment $\bm{M}^k = M_{i_1\dots i_k}\,dq^{i_1}\otimes \dots \otimes dq^{i_k}$ is a symmetric tensor whose components are given by the momentum space integrals
\begin{align*}
M_{i_1\dots i_{k}}(q) = \int p_{i_1}\dots p_{i_k}\,f(q,p)\,dp,\quad i_\ell=1,\dots,N,\quad k=0,1,2,\dots.
\end{align*}
In general, the evolution equation for the degree-$k$ fluid moment couples to the degree-$(k+1)$ fluid moment, a phenomenon that epitomizes the so-called moment closure problem. For example, evolution equations for the first three fluid moments $(M^0,M^1,M^2)=(n,\bm{P},\mathbb{S})$ are
\begin{gather}
\partial_tn + m^{-1}\nabla\cdot \bm{P} = 0\label{density_eqn}\\
\partial_t\bm{P} + m^{-1}\nabla\cdot\mathbb{S} = -e\,n\,\nabla{V}\label{momentum_eqn}\\
\partial_t\mathbb{S} + m^{-1}\nabla\cdot\mathcal{T} = -e(\nabla{V}\bm{P} + \bm{P}\nabla{V}),\label{momentum_flux_eqn}
\end{gather}
where $\mathcal{T} = M^3$ denotes the degree-$3$ fluid moment.
Any prescription for expressing the degree-$(k+1)$ fluid moment as a function of lower-degree fluid moments is known as a fluid closure. Upon introducing a fluid closure, the rate of change of the first $k$ fluid moments becomes a function of only those first $k$ fluid moments. Thus, it becomes feasible to solve a finite collection of hydrodynamic equations to find time evolution of the first $k$ fluid moments. Of course, there is no guarantee that moment evolution determined by a fluid closure bears any relation to moment evolution determined by the parent kinetic model, unless the fluid closure is chosen wisely.

%The quality of a fluid closure should be judged along several dimensions. The most important dimension is accuracy; moment evolution implied by the closure should come as close as possible to moment evolution implied by the kinetic model. Evaluating a closure along this dimension typically presents serious difficulties due to the high cost of extracting predictions from the kinetic model. Structure preservation defines the next most important dimension. Does the closure model enjoy qualitative properties that any ``exact" closure must enjoy \emph{a priori}? For example, since the Vlasov-Poisson system enjoys an energy conservation law, a good fluid closure should also enjoy an energy conservation law. Evaluating closures along this dimension usually presents relatively few difficulties. However, improving a given closure's structure-preserving properties presents serious challenges in general.

Traditional methods for developing fluid closures involve some combination of asymptotic expansions\cite{Gorban_2004,MacKay_2004,Gorban_2018,Burby_Klotz_2020} and intuitive guesswork. More recently, advances in machine learning\cite{Han_2019,Li_2023} have accelerated the development of data-driven fluid closures, as well as other data-driven reduced-order models\cite{Hesthaven_PIC_2023}. But fluid closures built using black-box neural networks are known to deliver poor results when employed over large simulation time intervals. There is therefore a growing effort to incorporate inductive bias into data-driven fluid closures, with the hope of improving long-term predictions. For example, Huang \emph{et al.}\cite{Huang_Christlieb_II_2021,Huang_Christlieb_III_2023} and Porteous \emph{et al.}\cite{Hauck_2021} developed variable-moment neural-network-based fluid closures that automatically enforce hyperbolicity of the ensuing moment evolution equations.

This Article presents a framework for building variable-moment fluid closures that automatically enjoy an important structural property inherent to non-dissipative kinetic plasma models: Hamiltonian structure. The framework involves two essential components. (1) Specification of a Poisson bracket for the first $m_0$ fluid moments, where $m_0$ is any integer greater than $1$. (2) Specification of the plasma system energy as a function of the first $m_0$ fluid moments. The form of the Poisson bracket only depends on the the number of moments $m_0$ included in a desired closure. On the other hand, the form of the  energy function, or Hamiltonian $\mathcal{H}$, is free to specify. 
%Each possible choice of $\mathcal{H}$ corresponds to a different fluid closure. 
Parameterizing $\mathcal{H}$ using a neural network as in Hamiltonian neural networks\cite{Greydanus_2019}, or a library of candidate terms\cite{Messenger_2021,Messenger_2022}, offers a novel path toward data-driven closure modeling for non-dissipative kinetic plasmas. As presented here, the framework applies specifically to fluid closures of the Vlasov-Poisson system in $N$ space dimensions. However, similar ideas apply to other non-dissipative kinetic plasma models, such as the Vlasov-Maxwell system\cite{Morrison_VM_1980,MW_1982,Morrison_proc_1982}.

The framework still applies in the presence of dissipation caused by particle collisions since a dissipative fluid closure should respect the metriplectic\cite{Morrison_met_1986} structure of collisional kinetic theory. Notably, metriplectic dynamical systems always decompose as the sum of a Hamiltonian part and a dissipative part; the framework developed here applies to the Hamiltonian part. 

The Hamiltonian structure underlying the \emph{infinite} hierarchy of fluid moments for the Vlasov-Poisson system was first identified by Gibbons\cite{Gibbons_1981}.
However, finding true Poisson brackets for finite, but arbitrarily large collections of fluid moments for the Vlasov-Poisson system represents an outstanding theoretical challenge.   The Poisson bracket for the closure framework presented here provides the first complete solution. In the context of slab drift kinetic models, Tassi\cite{Tassi_2015,Tassi_2022,Tassi_2023} previously found Hamiltonian fluid moment closures with any number of fluid moments. But the methods underlying this result fail to extend to other dissipation-free kinetic models in any obvious manner. In the Vlasov-Poisson context, Holm and Tronci\cite{Holm_Tronci_2009} found Hamiltonian fluid moment closures that include arbitrarily-large collections of fluid moments that \emph{do not} include the density moment. Otherwise, the best available partial solutions for Vlasov-Poisson work in one space dimension $(N=1)$, where the tensorial nature of fluid moments does not appear, and fall into two categories. The first category contains Hamiltonian moment closures derived by way of direct imposition of the Jacobi identity --- Jacobi closures. The second category contains closures based on exact solutions of the parent kinetic model that may be parameterized by fluid moments --- exact reductions.

At a schematic level, derivations of Jacobi closures work by first substituting an unknown closure relation into the Poisson bracket for the parent kinetic model, and then finding conditions the closure relation needs to satisfy in order for the closure bracket to satisfy the Jacobi identity. These conditions take the form of nonlinear partial differential equations (PDEs) that, in general, do not admit solutions or cannot be solved in closed form. However, there are notable exceptions. Tassi\cite{Tassi_2014} studied Jacobi closures for drift kinetics that include two moments. Chandre and Shadwick\cite{Chandre_Shadwick_2022} found a Jacobi closure that accommodates all fluid moments with degree $k < 5$ and ensures stability of uniform solutions. This extends the earlier result\cite{PCMT_2015} of Perin \emph{et al.} that also found Hamiltonian closures with $k<5$, but without stability guarantees. The feasibility of finding Jacobi closures was first established\cite{PCMT_2014} by Perin \emph{et al.} through a Jacobi-closure derivation of the earliest Hamiltonian fluid model due to Morrison and Greene\cite{Morrison_Greene_1980}. Jacobi closures involving moments with degree $k<2$ are automatic; the algebra of functionals that depend only on density and momentum density closes under the Vlasov-Maxwell Poisson bracket. The same algebraic closure property is not automatic for models that include moments with degrees $k\geq 2$. de Guillebon and Chandre\cite{Guillebon_Chandre_2012} dramatically highlighted this last point by revealing an earlier fluid moment model\cite{STE_2004} as a false-positive; a claimed Hamiltonian fluid closure including moments with degree $k= 2$ that was provably not Hamiltonian. Jacobi closures that include fluid moments of arbitrarily large degree have not yet been discovered.

Derivations of Hamiltonian moment closures based on exact reductions work like the classical variation of parameters method. After postulating a parametric form for the distribution function $f$, evolution equations for the parameters are derived that ensure a time-evolving parametric $f$ solves the kinetic equation exactly. This procedure only works for very special parametric $f$. For example, a distribution function of the form $f(q,p) = n(q)\,\delta(p-\bm{P}(q)/n(q))$, where $n$ is a positive function and $\bm{P}$ is a vector field, solves the Vlasov-Poisson system if and only if $(n,\bm{P})$ satisfies the Euler equations with a particular barotropic equation of state --- the exact reduction. The parameters for this $\delta$-closure are just the fluid moments with degrees $k<2$; the Poisson bracket among these parameters induced by the Vlasov-Poisson bracket recovers the $k<2$ closure bracket mentioned above. 
%Well-known exact reductions include linear combinations of $\delta$-functions [CITE] and multiple waterbags [CITE]. 
Perin \emph{et al.}\cite{PCMT_2015_waterbag} initiated a program aimed at finding moment closures with arbitrarily high-degree moments by re-parameterizing a multiple waterbag exact reduction. (Perin \emph{et al.}\cite{Perin_DK_2016} also studied waterbags in the drift kinetic context.) Simple explicit formulas express the fluid moments of multiple waterbags as polynomials in waterbag parameters. Whenever these polynomial relationships can be inverted to give waterbag parameters in terms of fluid moments a Hamiltonian moment closure results. However, due to the general insolubility of large-degree polynomial equations, the inversion has only been demonstrated for a single waterbag. Like Jacobi closures, Hamiltonian closures based on exact reductions have yet to accommodate fluid moments of arbitrarily large degree.

The Hamiltonian moment closures presented here fall into the exact reduction category. The relevant parametric form for $f$ is given by
\begin{align}
f(q,p)&=n(q)\,\delta(p+d\psi(q)) - \mu_{i_1}(q)\,\partial_{p_{i_1}}\delta(p+d\psi(q))\nonumber\\
&+ \frac{1}{2}\,\mu_{i_1\,i_2}(q)\,\partial^2_{p_{i_1\,i_2}}\delta(p+d\psi(q)) + \dots + \frac{(-1)^{m_0}}{m_0!}\mu_{i_1\dots i_{m_0}}(q)\,\partial^{m_0}_{p_{i_1}\dots p_{i_{m_0}}}\delta(p+d\psi(q)). \label{f_ansatz_basic}
\end{align}
Here, $n$ denotes the density, the gradient of the scalar function $\psi$ determines the center of the distribution function, and $\mu_{i_1\dots i_k}$ denotes the degree-$k$ centered fluid moment. Section \ref{sec:exact_reduction} contains a detailed exposition. A simple upper-triangular transformation, Eq.\,\eqref{binomial_formula}, relates the first $m_0$ centered fluid moments to the first $m_0$ standard fluid moments. Hamiltonian closures based on this form of $f$ therefore accommodate arbitrarily high-degree moments by increasing the integer $m_0$. This feature of the closures, together with their applicability in every space dimension $N$, sets them apart from all previous Hamiltonian fluid closures that currently appear in the literature. 

The ansatz \eqref{f_ansatz_basic} can reproduce the first $m_0$ moments of \emph{any} distribution. The Poisson bracket associated with this exact reduction may therefore be used to model the first $m_0$ moments of any distribution. In particular, it is not necessary to assume that the distributions modeled by the new closures are exactly of the form \eqref{f_ansatz_basic}. This observation forms the basis for the proposed closure modeling approach, described with greater detail in Section \ref{sec:general_scheme}: fix the Poisson bracket and incorporate all modeling choices into the form of the Hamiltonian $\mathcal{H}$. However, it is important to keep in mind that the theory developed here only guarantees the existence of closures that apply to $f$ close to the form \eqref{f_ansatz_basic}. Applications of the new closures to, e.g., bump-on-tail distributions should be considered exploratory efforts.

Demonstrating that \eqref{f_ansatz_basic} comprises an exact reduction and uncovering the corresponding Poisson bracket structure for general $m_0$ and $N$ requires considerable bookkeeping technology. Fortunately, Scovel and Weinstein\cite{Scovel_Weinstein_1994} developed the necessary general machinery in the context of distributions $f$ localized near a point $z = (q,p)$ in single-particle phase space. (Channell\cite{Channell_1995} provides additional insights into elements of the Scovel-Weinstein construction.) Such distributions, which are relevant to self-consistent beam dynamics in particle accelerators, also give rise to a moment closure problem, where the moments are redefined as the phase space integrals
\begin{align*}
m^{\nu_1\dots \nu_k} = \int z^{\nu_1}\dots z^{\nu_k}f(z)dz, \quad \nu_\ell=1,\dots,2N.
\end{align*}
In order to find a finite-dimensional Poisson bracket that describes dynamics of the first $m_0$ phase space moments, Scovel-Weinstein developed a reduction strategy that can be applied to any Lie-Poisson Hamiltonian system whose associated Lie algebra $\mathfrak{g}$ decomposes as a sum of subalgebras $\mathfrak{g} = \mathfrak{b} + \mathfrak{s}$. (The Lie bracket on $\mathfrak{g}$ need not factor as the sum of brackets from $\mathfrak{b}$ and $\mathfrak{s}$.) For the beam problem, $\mathfrak{g}$ comprises the phase space functions with multiplication provided by the usual single-particle Poisson bracket, $\mathfrak{b}$ contains the affine phase space functions, and $\mathfrak{s}$ is the collection of phase space functions that vanish at $(q,p) = (0,0)$, together with their first partial derivatives. For fluid moments, the Lie algebra $\mathfrak{g}$ remains the same, but the subalgebras $\mathfrak{b},\mathfrak{s}$ differ substantially. Section \ref{sec:lie} summarizes the Lie theory necessary to apply Scovel-Weinstein theory to fluid moments, and in particular defines the subalgebras $\mathfrak{b}$ and $\mathfrak{s}$. Section \ref{sec:scovel} precisely states the elements of Scovel-Weinstein theory that enable subsequent development of the new closure method in Sections \ref{sec:PB}, \ref{sec:exact_reduction}, and \ref{sec:general_scheme}.

The following notation is used in what follows without further comment. If $V=\mathbb{R}^d$ is any finite-dimensional vector space then $C(V)$ contains smooth functions on $V$ with mild growth (e.g. at most polynomial grown) at infinity. The dual space $C^*(V)$ contains rapidly decaying distributional functions on $V$. The duality pairing between $\kappa\in C(V)$ and $\rho\in C^*(V)$ is $\langle\rho,\kappa \rangle = \int_{V} \kappa(v)\,\rho(v)\,dv$, where $dv$ denotes the Lebesgue measure on $\mathbb{R}^d$.

\section{Lie theory for the contactomorphism group\label{sec:lie}}
A complete understanding of the Hamiltonian moment closures described in this Article requires familiarity with some of the rich structure of the \textbf{contactomorphism group} $G$. This Section aims to provide a self-contained introduction to $G$ and its properties with greatest relevance to fluid moments. For this Section and Section \ref{sec:scovel} only, readers should be familiar with the basic mathematical theory of Lie groups at the level of Chapter 4 in the book\cite{Abraham_2008} by Abraham and Marsden. Marsden and Weinstein\cite{MW_1982} were first to recognize that $G$ is the Lie group that integrates the Lie algebra underlying the Vlasov equation. Subsequently, Gay-Balmaz and Tronci\cite{BT_2012} used structural properties of contactomorphisms to show that the integrable Bloch-Iserles system is geodesic flow on a subgroup of $G$. Gay-Balmaz and Tronci\cite{BT_2020,BT_2022} also provide additional exposition concerning the contactomorphism group as part of their ongoing efforts in quantum-classical hybrid modeling.

Let $Q\ni q$ be an $N$-dimensional vector space representing the configuration space for an $N$-degree of freedom Hamiltonian system. The corresponding momentum phase space is $P =Q\times Q^* \ni (q,p)$, where $Q^*$ denotes the dual vector space to $Q$. The cotangent projection map is $\pi:P\rightarrow Q:(q,p)\mapsto q$. Each $p\in Q^*$ is a linear mapping $p:Q\rightarrow\mathbb{R}$. Application of the linear map $p\in Q^*$ to a vector $q\in Q$ will be denoted $\langle p,q\rangle\in\mathbb{R}$. The \textbf{canonical $1$-form} $\vartheta$ on $P$ is defined by $\vartheta = \langle p,dq\rangle$. The corresponding \textbf{canonical symplectic form} is the $2$-form $\omega = -d\vartheta$ given by exterior differentiation of $\vartheta$. If $h:P\rightarrow\mathbb{R}$ is any smooth function on $P$, its partial derivatives $\delta h/\delta q$ and $\delta h/\delta p$ are defined by the directional derivative formula
\begin{align*}
\frac{d}{d\epsilon}\bigg|_0 h(q+\epsilon\,\delta q,p+\epsilon\,\delta p) = \left\langle \frac{\delta h}{\delta q} ,\delta q\right\rangle +  \left\langle\delta p, \frac{\delta h}{\delta p} \right\rangle.
\end{align*}
Note that the value of $\delta h/\delta q$ at a particular $(q,p)\in P$ lies in $Q^*$, while $\delta h/\delta p$ lies in $Q$. The Poisson bracket associated with the canonical symplectic form $\omega$ is given on pairs of functions $f,g:P\rightarrow\mathbb{R}$ by
\begin{align*}
\{f,g\} = \left\langle \frac{\delta f}{\delta q},\frac{\delta g}{\delta p} \right\rangle - \left\langle \frac{\delta g}{\delta q},\frac{\delta f}{\delta p} \right\rangle.
\end{align*}

A (strict) \textbf{contactomorphism} $g = (\Phi,\varphi)$ is a pair comprising a diffeomorphism $\Phi:P\rightarrow P$ and a smooth function $\varphi:P\rightarrow\mathbb{R}$ that together satisfy the identity
\begin{align*}
d\varphi + \Phi^*\vartheta = \vartheta.
\end{align*}
The product of a pair of contactomorphisms $g_1,g_2$, denoted $g_1g_2$, is 
\begin{align*}
g_1g_2 = (\Phi_1\circ \Phi_2, \Phi_2^*\varphi_1 + \varphi_2).
\end{align*}
When equipped with this product, the set of all contactomorphisms comprise a group known as the (strict) contactomorphism group.

% The product $g_1g_2$ is itself a contactomorphism because 
% \begin{align*}
% d(\Phi_2^*\varphi_1 + \varphi_2) + (\Phi_1\circ \Phi_2)^*\vartheta = d(\Phi_2^*\varphi_1 + \varphi_2)+ \Phi_2^*(\vartheta - d\varphi_1) = d(\Phi_2^*\varphi_1 + \varphi_2) + \vartheta - d(\Phi_2^*\varphi_1+\varphi_2) = \vartheta.
% \end{align*}
% The product operation on contactomorphisms is associative, but not commutative. Moreover, if $g$ is a contactomorphism its multiplicative inverse is the contactomorphism
% \begin{align*}
% g^{-1} = (\Phi^{-1},-(\Phi^{-1})^*\varphi),
% \end{align*}
% and $e = (\text{id},0)$ is a multiplicative identity. It follows that the set of all contactomorphisms $G$ is a group --- the contactomorphism group.

The contactomorphisms $b = (\Phi,\varphi)$ such that
\begin{align*}
\Phi(q,p) = \bigg(q,p - \frac{\delta \psi}{\delta q}\bigg),\quad \varphi(q,p) = \psi(q),\quad \psi\in C(Q).
\end{align*}
form a subgroup $B$, known as the group of \textbf{fiber translations}. Note that if $\tau_\psi: P\rightarrow P$ denotes the mapping $(q,p)\mapsto (q,p - \delta\psi/\delta q)$, then every fiber translation may be written as $b = (\tau_\psi,\pi^*\psi)$ for some unique $\psi$. The mapping $I_B :  B\rightarrow C(Q): (\tau_\psi,\pi^*\psi)\mapsto \psi$, from $B$ to the Abelian group of smooth functions on configuration space is a group isomoprhism.

Let $Z = \{(q,p)\in P\mid p = 0\}$ denote the \textbf{zero section} in phase space. The submanifold $Z$ contains the static system states. The contactomorphisms $s= (\Phi,\varphi)$ such that $\Phi(Z) = Z$ and $\varphi(Z)=0$ form a Lie subgroup $S\subset G$. In this work, $S$ will be referred to as the \textbf{group of isostatic contactomorphisms}. 
% This terminology follows from the fact that each contactomorphism in $S$ preserves the collection of static system states. 
The isostatic contactomorphisms complement the fiber translations in the following sense.

\begin{theorem}\label{G_decomp_thm}
The product mapping 
\begin{align*}
m:B\times S\rightarrow G: (b,s)\mapsto bs,
\end{align*}
restricts to a diffemorphism from a neighborhood of $(e,e)\in B\times S$ to a neighborhood of $e\in G$. In particular, for each $g\in G$ sufficiently close to the identity there is a unique fiber translation $b$ and a unique isostatic contactomorphism $s$ such that $g=bs$.
\end{theorem}
\begin{proof}
The proof proceeds by constructing the inverse of $m_0$ on an appropriate neighborhood $U_G$ of the identity in $G$. In particular, given $g\in G$ the key step is finding a cooresponding $b = b(g)$ such that $s(g) = b(g)^{-1}g\in S$. The inverse of $m_0$ is then $m^{-1}(g)=(b(g),s(g))$, since $m(b(g),s(g)) = b(g)s(g) = b(g)(b(g)^{-1}g) = g$. 

Observe that the zero section $Z$ is a graph over $Q$. Let $U_G\subset G$ denote the set of contactomorphisms $g = (\Phi,\varphi)$ such that $\Phi(Z)$ is a graph over $Q$. Given $g\in U_G$, the graph $\Phi(Z)$ is a Lagrangian submanifold of $P$. There is therefore a function $\psi:Q\rightarrow \mathbb{R}$ on configuration space such that $\Phi(Z)$ is the graph of (minus) the differential of $\psi$, 
\[
\Phi(Z) = \{(q,p)\in P\mid p = -\delta \psi/\delta q\}.
\]
In other words, $\psi$ is (minus) a generating function for the Lagrangian submanifold $\Phi(Z)$. A generating function $\psi$ for $\Phi(Z)$ determines a fiber translation $b = (\tau_\psi,\pi^*\psi)$, which very nearly furnishes the $b = b(g)$ mentioned above. However, there is more to do because $(\tau_\psi,\pi^*\psi)^{-1}g$ is not necessarily in $S$. Moreover, the fiber translation $(\tau_\psi,\pi^*\psi)$ is not uniquely determined by $g$. In fact, $\psi^\prime = \psi + c$, $c\in \mathbb{R}$, also generates $\Phi(Z)$, but determines a different fiber translation $(\tau_{\psi^\prime},\pi^*\psi^\prime) = (\tau_\psi,\pi^*\psi + c)$.

The following argument shows that for each $g\in U_G$ there is a unique generating function $\psi = \psi(g)$ for $\Phi(Z)$ such that $s = (\tau_\psi,\pi^*\psi)^{-1}g\in S$. Let $\iota_Q:Q\rightarrow P: q\mapsto (q,0)$ denote the map that embeds $Q$ in $P$ as $Z$, and set $ (\Phi_s,\varphi_s)= s = (\tau_\psi^{-1}\circ\Phi,-\Phi^*\pi^*\psi + \varphi)$. As mentioned already, there exists a (non-unique) $\psi$ that generates $\Phi(Z)$. Since $\Phi(Z)$ is the graph of $-\delta \psi/\delta q$, the mapping $\tau_{\psi}^{-1}$ maps the Lagrangian submanifold $\Phi(Z)$ back to $Z$, i.e. $\Phi_s(Z) = \tau_{\psi}^{-1}(\Phi(Z)) = Z$. Equivalently, there is a diffeomorphism $\eta_s:Q\rightarrow Q$ such that $\Phi_s\circ\iota_Q = \iota_Q\circ \eta_s$. Since $s$ is a contactomorphism, $d\varphi_s + \Phi_s^*\vartheta = \vartheta$. Pulling back this identity along $\iota_Q$ implies
\[
0 = \iota_Q^*\vartheta = \iota_Q^*(d\varphi_s + \Phi_s^*\vartheta) = d\iota_Q^*\varphi_s + (\Phi_s\circ\iota_Q)^*\vartheta = d\iota_Q^*\varphi_s + (\iota_Q\circ\eta_s)^*\vartheta =d\iota_Q^*\varphi_s + \eta_s^*\iota_Q^*\vartheta = d\iota_Q^*\varphi_s. 
\]
This means $\varphi_s(q,0) = \varphi(q,0) - \psi(\eta_s(q)) $ is constant as a function of $q$. Eliminate the ambiguity in the definition of $\psi$ by requiring that $\varphi_s(q,0) = 0$. Once this is done, $b=b(g) = (\tau_{\psi(g)},\pi^*\psi(g))$ is uniquely determined by $g$ and $b^{-1}g\in S$, as desired. The inverse of $m_0$ on $U_G$ is now $m^{-1}(g) = (b(g),s(g))$, with $s(g) = b(g)^{-1}g$.
\end{proof}

The Lie algebra $\mathfrak{g} = T_eG$ of \textbf{infinitesmial contactomorphisms} comprises pairs $(X,\Sigma)$, where $X$ is a vector field on $P$, $\Sigma:P\rightarrow\mathbb{R}$ is a function on $P$, and $d\Sigma +\mathcal{L}_X\vartheta = 0.$
% \begin{align*}
% d\Sigma +\mathcal{L}_X\vartheta = 0.
% \end{align*}
The adjoint action of $G$ and Lie bracket on $\mathfrak{g}$ are given by
\begin{align*}
\text{Ad}_g(X,\Sigma) &= (\Phi_*X,\Phi_*(\Sigma + \mathcal{L}_X\varphi)),\quad g = (\Phi,\varphi),\\
[(X_1,\Sigma_1),(X_2,\Sigma_2)] &= (-[X_1,X_2],-\mathcal{L}_{X_1}\Sigma_2 +\mathcal{L}_{X_2}\Sigma_1).
\end{align*}
% The Lie bracket on $\mathfrak{g}$ is therefore
% \begin{align*}
% [(X_1,\Sigma_1),(X_2,\Sigma_2)] = (-[X_1,X_2],-\mathcal{L}_{X_1}\Sigma_2 +\mathcal{L}_{X_2}\Sigma_1).
% \end{align*}
The linear map $I_{\mathfrak{g}}:\mathfrak{g}\rightarrow C(P):(X,\Sigma)\mapsto \Sigma + \vartheta(X)$, from infinitesimal contactomorphisms to phase space functions under Poisson bracket
% satisfies
% \begin{align*}
% \{I_{\mathfrak{g}}(X_1,\Sigma_1),I_{\mathfrak{g}}(X_2,\Sigma_2)\} = I_{\mathfrak{g}}[(X_1,\Sigma_1),(X_2,\Sigma_2)].
% \end{align*}
% Moreover $I_{\mathfrak{g}}$ is invertible with inverse $I_{\mathfrak{g}}^{-1}:h\mapsto (X_h,h - \vartheta(X_h))$. It follows that $I_{\mathfrak{g}}$ 
is a Lie algebra isomorphism. This observation recovers the well-known fact that the contactomorphism group integrates the Lie algebra of phase space functions equipped with the Poisson bracket. 
% The isomorphism $I_{\mathfrak{g}}$ allows expressing the adjoint action on $\mathfrak{g}$ as an equivalent $G$-action on $C(P)$ given by
% \begin{align*}
% \text{Ad}_gh = I_{\mathfrak{g}}\text{Ad}_gI_{\mathfrak{g}}^{-1}h = \Phi_*h,\quad h\in C(P),\quad g= (\Phi,\varphi)\in G.
% \end{align*}

The Lie subalgebra $\mathfrak{b} = T_{e}B\subset\mathfrak{g}$ of \textbf{infinitesimal fiber translations} comprises pairs $(X,\Sigma)\in \mathfrak{g}$, where $\Sigma = \pi^*\sigma$ for some $\sigma\in C(Q)$ and $X = X_{\pi^*\sigma}$. The Lie bracket on $\mathfrak{b}$ vanishes because $B$ is Abelian. When identifying $\mathfrak{g}$ with phase space functions \emph{via} $I_{\mathfrak{g}}$, $\mathfrak{b}$ corresponds to $p$-independent functions.
% :
% \begin{align*}
% I_{\mathfrak{g}}(X,\Sigma) = \Sigma + \vartheta(X_\sigma) = \pi^*\sigma + \left\langle p,\frac{\delta \pi^*\sigma}{\delta p}\right\rangle= \pi^*\sigma.
% \end{align*}
Denote the image Lie subalgebra of $\mathfrak{b}$ in $(C(P),\{\cdot,\cdot\})$ as 
\begin{align*}
C_{\mathfrak{b}}(P) = \{h\in C(P)\mid \delta h/\delta p = 0\}.
\end{align*}

The Lie subalgebra $\mathfrak{s} = T_{e}S\subset\mathfrak{g}$ of \textbf{isostatic infinitesimal contactomorphisms} comprises pairs $(X,\Sigma)\in\mathfrak{g}$ such that $X$ is tangent to $Z$ and $\Sigma(Z) = 0$. (In fact, the latter condition implies the former, and so can be taken as the sole property characterizing $\mathfrak{s}$ as a subspace of $\mathfrak{g}$.) The image Lie subalgebra $I_{\mathfrak{g}}(\mathfrak{s})$ in $(C(P),\{\cdot,\cdot\})$ is
\begin{align*}
C_{\mathfrak{s}}(P) = \{h\in C(P)\mid h(Z) = 0\}.
\end{align*}

Theorem \ref{G_decomp_thm} implies that, as a vector space, the Lie algebra of infinitesimal contactomorphisms is the direct sum of $\mathfrak{b}$ and $\mathfrak{s}$, $\mathfrak{g} = \mathfrak{b} + \mathfrak{s}$. The corresponding statement in $C(P)$, namely $C(P) = C_{\mathfrak{b}}(P) + C_{\mathfrak{s}}(P)$, may be understood more simply as a byproduct of Taylor's theorem with remainder. 
% Indeed, if $h\in C(P)$ then Taylor's theorem with remainder implies
% \begin{align*}
% h(q,p) = h(q,0) + \langle p,h_R(q,p)\rangle,
% \end{align*}
% where $h_R(q,p)\in Q$ is a smooth function of $(q,p)$. The first term in this sum sits in $C_{\mathfrak{b}}(P)$, while the second term defines an element of $C_{\mathfrak{s}}(P)$.

The dual space $\mathfrak{g}^*$ comprises linear functions $\mu : \mathfrak{g}\rightarrow \mathbb{R}$. Since $\mathfrak{g} = \mathfrak{b} + \mathfrak{s}$, the duals to $\mathfrak{b}$ and $\mathfrak{s}$ are subspaces of $\mathfrak{g}^*$ given by
\begin{align*}
\mathfrak{b}^* = \{\mu\in \mathfrak{g}^*\mid \langle\mu,\mathfrak{s} \rangle = 0\},\quad \mathfrak{s}^* = \{\mu\in \mathfrak{g}^*\mid \langle\mu,\mathfrak{b} \rangle = 0\}.
\end{align*}
Similarly, the dual space to $C(P)$, $C^*(P)$, contains linear functions ${f}:C(P)\rightarrow\mathbb{R}$, and the duals $C_{\mathfrak{b}}^*(P)$, $C_{\mathfrak{s}}^*(P)$, are the subspaces of $C^*(P)$ given by
\begin{align*}
C_{\mathfrak{b}}^*(P)=\{{f}\in C^*(P)\mid \langle {f},C_{\mathfrak{s}}(P)\rangle=0\},\quad C_{\mathfrak{s}}^*(P)=\{{f}\in C^*(P)\mid \langle {f},C_{\mathfrak{b}}(P)\rangle=0\}.
\end{align*}
The coadjoint action of $G$ on $\mathfrak{g}^*$ is given by $\text{Ad}^*_g:\mathfrak{g}^*\rightarrow\mathfrak{g}^*: (X^*,\varphi^*)\mapsto (\text{Ad}_{g^{-1}})^*(X^*,\varphi^*)$. The corresponding coadjoint action of $G$ on $C^*(P)$ is
\begin{align*}
\text{Ad}^*_g{f} = \Phi_*{f},\quad {f}\in C^*(P),\quad g=(\Phi,\varphi)\in G.
\end{align*}
Provided the space of phase space functions is restricted to smooth functions with slow growth at infinity, the dual $C^*(P)$ contains all of the typical smooth Vlasov densities. It includes more exotic functionals as well, such as Dirac $\delta$-functions and their derivatives.

\section{Scovel-Weinstein theory\label{sec:scovel}}
The (infinite-dimensional) Hamiltonian structure underlying the Vlasov-Poisson system is dual to the Lie algebra $C(P)$ of phase space functions under the standard (finite-dimensional) Poisson bracket. The problem of finding reduced models, in particular reduced fluid models, for Vlasov-Poisson may therefore be approached at the level of the algebra $C(P)$. Adopting this point of view requires addressing how the Lie algebra $C(P)$ might be ``reduced." For the purposes of this Article, any reduction should somehow encode the  objects dual to the first $m_0$ fluid moments, namely the phase space functions $h(q,p)$ that are polynomial in momentum $p$ with degree at most $m_0$.

Plausible reductions of $C(P)$ include proper Lie subalgebras: subspaces of $C(P)$ that close under Poisson bracket. For example, the phase space functions that only depend on the configuration variable $q$ (i.e. degree-$0$ momentum polynomials) comprise an Abelian Lie subalgebra dual to the space of density moments. More generally, the momentum polynomials with degree at most $1$ close under Poisson bracket because
\begin{align*}
\left\{\bigg(c_1(q)+c_1^i(q)p_i\bigg),\bigg(c_2(q) + c_2^i(q)p_i\bigg)\right\} = -\bigg(c_1^i(q)\partial_ic_2(q)-c_2^i(q)\partial_ic_1(q)\bigg) - \bigg(c_1^j(q)\partial_jc_2^i(q) - c_2^j(q)\partial_j c_1^i(q)\bigg)p_i.
\end{align*}
However, by itself, subalgebra reduction of $C(P)$ provides insufficient flexibility to accommodate fluid moments with degree between $0$ and $m_0$ when $m_0 > 1$. This follows from the failure of momentum polynomials with degree at most $2$ to close under Poisson bracket.

Quotient algebra reduction offers a more nuanced alternative to subalgebra reduction. Instead of attempting to encapsulate moments in a small subalgebra, quotient algebra reduction aims to find \emph{large} ideals $\mathcal{I}$ in $C(P)$ such that the quotient space $C_{\mathcal{I}}(P) = C(P)/\mathcal{I}$ encodes fluid moments. Here, ideals are subalgebras with an additional absorbing property: the bracket of functions in $\mathcal{I}$ with \emph{arbitrary} functions in $C(P)$ must be contained in $\mathcal{I}$. The subspace $\mathcal{I}$ being an ideal ensures that the quotient $C_{\mathcal{I}}(P)$ inherits a Lie algebra structure from $C(P)$. Unfortunately, as for subalgebra reduction, quotient algebra reduction cannot accommodate the fluid moments with degree between $0$ and $m_0$ by itself. As described in detail within the proof of Theorem \ref{moment_LA_thm}, it \emph{can} successfully encode moments with degree between $1$ and $m_0$.

The paper\cite{Scovel_Weinstein_1994} by Scovel and Weinstein devises an alternative Lie-theoretic reduction strategy that hybridizes subalgebra reduction and quotient algebra reduction. Developed originally for application to phase space moments, Scovel-Weinstein theory adapts to the fluid moment context in straightforward fashion. In fact, their methods apply any time a Lie group $G$ admits a decomposition by Lie subgroups $B,S\subset G$. Here, decomposition means that for each $g\in G$ there are unique $b\in B$ and $s\in S$ such that $g = bs$. In the context of fluid moments, the group $G$ is the contactomorphism group, $B$ is the group of fiber translations, and $S$ is the group of isostatic contactomorphisms. Each of these groups is described with some detail in Section \ref{sec:lie}. The technical heart of Scovel-Weinstein theory may be described as follows.

Suppose that $G$ is a Lie group that admits a decomposition by Lie subgroups $B,S\subset G$, $G = BS$. The Lie algebra $\mathfrak{g}$ then admits a corresponding direct sum decomposition $\mathfrak{g} = \mathfrak{b} + \mathfrak{s}$ by Lie subalgebras. Likewise, the dual space to $\mathfrak{g}$ admits the direct sum decomposition $\mathfrak{g}^* = \mathfrak{b}^* + \mathfrak{s}^*$. The \textbf{Scovel-Weinstein map} $\Gamma:(B\times\mathfrak{b}^*) \times \mathfrak{s}^*\rightarrow\mathfrak{g}^*$ is given by
\begin{align*}
\Gamma(b,b^*,s^*) = \text{Ad}^*_{b}(b^* + s^*).
\end{align*}
Here $\text{Ad}^*_g=(\text{Ad}_{g^{-1}})^*$ denotes the (left) coadjoint action of $G$ on $\mathfrak{g}^*$. Let $\{\cdot,\cdot\}_{T^*B}$ denote the Poisson bracket on $B\times \mathfrak{b}^*$ given by pulling back the canonical Poisson bracket on $T^*B$ along the left trivialization map. Let $\{\cdot,\cdot\}_{\mathfrak{s}^*}$ denote the Lie-Poisson bracket on $\mathfrak{s}^*$. Finally, let $\{\cdot,\cdot\}_{\mathfrak{g}^*}$ denote the Lie-Poisson bracket on $\mathfrak{g}^*$. The central result due to Scovel-Weinstein is the following. When $(B\times\mathfrak{b}^*) \times \mathfrak{s}^*$ is equipped with the evident product Poisson manifold structure, the Scovel-Weinstein map $\Gamma$ is a Poisson map. In other words, the direct product Poisson bracket is compatible with the Lie-Poisson bracket on $\mathfrak{g}$ according to
\begin{align*}
\{\Gamma^*f,\Gamma^*g\}_{T^*B} + \{\Gamma^*f,\Gamma^*g\}_{\mathfrak{s}^*} = \Gamma^*\{f,g\}_{\mathfrak{g}^*},\quad f,g\in C(\mathfrak{g}^*).
\end{align*}
Note that the space $(B\times\mathfrak{b}^*)\times\mathfrak{s}^*$ is larger that $\mathfrak{g}^*$ by the factor $B$. This extra variable $b$ is the price paid for ``untangling" the Lie-Poisson structure on $\mathfrak{g}^*$ in terms of standard Poisson structures associated with $\mathfrak{b}^*$ and $\mathfrak{s}^*$. Scovel-Weinstein advocate the following interpretation of $\Gamma$ and its source and target spaces. In the space $(B\times\mathfrak{b}^*)\times \mathfrak{s}^*$ the factor $\mathfrak{g}_0^* = \mathfrak{b}^*\times\mathfrak{s}^*$ represents a ``reference" copy of $\mathfrak{g}^*$ sitting at the identity $e$ in the group $B$. The mapping $\Gamma$ advects the reference space $\mathfrak{g}_0^*$ from $e$ to $b$.

This advection interpretation warrants further elaboration. Along a solution of the Vlasov-Poisson system, the distribution function $f(q,p)$ is advected by canonical transformations that prescribe the motion of individual particles in phase space. These canonical transformations naturally reside with the contactomorphism group $G$. The Scovel-Weinstein map $\Gamma$ may therefore be interpreted as a general way of identifying special system motions advected by a subgroup $B\subset G$.

Now suppose that the Lie subalgebra $\mathfrak{s}$ contains an ideal $\mathcal{I}\subset \mathfrak{s}$. Note that ideals in $\mathfrak{s}$ are not necessarily ideals in $\mathfrak{g}$! The \textbf{quotient algebra} $\mathfrak{s}_{\mathcal{I}} = \mathfrak{s}/\mathcal{I}$ then inherits a Lie algebra structure from $\mathfrak{s}$ by requiring that the quotient map $\Pi_{\mathcal{I}}:\mathfrak{s}\rightarrow\mathfrak{s}_{\mathcal{I}}$ is a Lie algebra homomorphism. If $s_1+\mathcal{I},s_2+\mathcal{I}\in\mathfrak{s}_{\mathcal{I}}$ are elements of the quotient space their Lie bracket is
\begin{align*}
[s_1+\mathcal{I},s_2+\mathcal{I}]= [s_1,s_2]+\mathcal{I}.
\end{align*}
Scovel-Weinstein highlight that since $\Pi_{\mathcal{I}}$ is a Lie algebra homomorphism, the dual map $\Pi_{\mathcal{I}}^*:\mathfrak{s}^*_{\mathcal{I}}\rightarrow \mathfrak{s}^*$ is a Poisson map. Therefore the map
\begin{align*}
\text{id}\times \Pi_{\mathcal{I}}^*: (B\times \mathfrak{b}^*)\times \mathfrak{s}^*_{\mathcal{I}}\rightarrow (B\times \mathfrak{b}^*)\times \mathfrak{s}^*,
\end{align*}
is Poisson as well, relative to the product Poisson manifold structures on the source and target spaces. It follows that there is a second Poisson map related to $\Gamma$ \emph{via} the quotient algebra construction. The \textbf{reduced Scovel-Weinstein map} $\Gamma_{\mathcal{I}}$ is the Poisson map
\begin{align*}
\Gamma_{\mathcal{I}}: (B\times \mathfrak{b}^*)\times \mathfrak{s}^*_{\mathcal{I}}\rightarrow\mathfrak{g}^*
\end{align*}
given as the composition $\Gamma_{\mathcal{I}} = \Gamma\circ [\text{id}\times \Pi_{\mathcal{I}}^*]$. Theorem \ref{thm:sw_map_for_moments} below represents a restatement of the Poisson map property for $\Gamma_{\mathcal{I}}$ within the fluid moment context.

\section{Poisson bracket on the space of truncated fluid moment vectors\label{sec:PB}}
% 1) exhibit Lie bracket on \mathfrak{s}_k
This Section uses Poisson-geometric and Lie-theoretic methods to identify a Poisson bracket on a space that encodes density, fluid momentum, pressure, and higher fluid moments up to some truncation order $m_0$. The precise relationship between the Poisson bracket constructed here and the Lie Poisson bracket on the space of distribution functions $f\in C^*(P)$ will be described in Section \ref{sec:exact_reduction}.

A degree-$k$ \textbf{fluid moment} $M^k$, $k\in\{0,1,2,\dots\}$, is a symmetric, covariant, degree-$k$ tensor field on configuration space $Q$. (Geometrically, fluid moments are naturally tensor densities\cite{Tronci_geom_2008} instead of ordinary tensors, but here these two kinds of objects are identified using the standard flat metric on configuration space.) If $\{q^i\}_{i=1,\dots,N}$ denotes the standard coordinate system on $Q=\mathbb{R}^N$, $M^k$ may be written as
\begin{align*}
M^k  = M_{i_1\,\dots\,i_k}\,dq^{i_1}\otimes \dots \otimes dq^{i_k},\quad M_{i_1\,\dots\,i_k}\in C^*(Q).
\end{align*}
Any phase space distribution $f=f(q,p)$ that decays sufficiently fast as $|p|\rightarrow\infty$ determines fluid moments of each degree by the formula
\begin{align*}
M_{i_1\dots i_k}(q) = \int p_{i_1}\dots p_{i_k}\,f(q,p)\,dp.
\end{align*}
Conversely, under certain conditions a sequence of fluid moments $M^k$, $k=0,\dots,\infty$, determines a distribution $f$. This is true, for example, when the moment generating function associated with the sequence $M^k$ exists. A \textbf{fluid moment vector} with degree between $n\geq 0$ and $m\geq n$ is a tuple $\bm{M} = (M^n,M^{n+1},\dots,M^{m})$ whose $i^{\text{th}}$ entry is a fluid moment of degree $i$. Denote the set of all such fluid moment vectors $\mathcal{M}_n^{m_0}$. Note that $\mathcal{M}_{k}^k$ is the vector space of degree-$k$ fluid moments, and that
\begin{align*}
\mathcal{M}_{n}^{m_0} = \bigoplus_{k=n}^{m_0} \mathcal{M}_{k}^k.
\end{align*}
Also note that elements of $\mathcal{M}_{0}^0$ are just fluid densities $n\in C^*(Q)$.

Although specifying moments of each degree $k=0,\dots,\infty$ formally specifies a unique phase space distribution, specifying only the moments with degree $k\leq m_0$ does not. Additional data must be supplied to recover a distribution function from a fluid moment vector in $\mathcal{M}_0^{m_0}$. In this work that additional data will comprise a scalar configuration space function $\psi\in C(Q)$ whose derivative $d\psi(q)$ specifies the ``center" of the distribution $f(q,p)$ at $q$. The precise goal of this Section is constructing a Poisson bracket on the space 
\begin{align*}
P_{m_0} = C(Q)\times \mathcal{M}_0^{m_0} = C(Q)\times C^*(Q) \times \mathcal{M}_1^{m_0}\ni (\psi,n,\bm{\mu}),
\end{align*}
for any non-negative integer $m_0$. Here the symbol $\bm{\mu}$ is used in place of $\bm{M}$ to emphasize that the moment vector $\bm{\mu}$ contains \emph{centered moments}.

The Poisson bracket on $P_{m_0}$ originates from a Lie algebra structure on the objects dual to fluid moment vectors: fluid comoment vectors. A \textbf{fluid comoment} $C^k$ of degree $k\in\{0,1,2,\dots\}$ is a symmetric, contravariant, degree-$k$ tensor field on configuration space $Q$. If $\{e_i\}_{i\in\{1,\dots,N\}}$ denotes the standard basis for $Q$, evey degree-$k$ fluid comoment may be expressed using index notation as
\begin{align*}
C^k = C^{i_1\,i_2\,\dots\, i_k} e_{i_1}\otimes e_{i_2}\otimes\dots\otimes e_{i_k},\quad C^{i_1\,i_2\,\dots\, i_k}\in C(Q),
\end{align*}
where the array of component functions $C^{i_1\,i_2\,\dots\, i_k}$ is symmetric under index permutation.
A \textbf{fluid comoment vector} with degree between $n\geq 0$ and $m\geq n$ is a tuple $\bm{C} = (C^n,C^{n+1},\dots,C^{m})$ whose $i^{\text{th}}$ entry is a fluid comoment of degree $i$. The vector space of all fluid comoment vectors with degree between $n$ and $m_0$ is $\mathcal{C}_n^{m_0}$. Fluid comoments with degrees $k=0,1,2$ will be denoted using the special symbols $\sigma, \bm{w}, \mathbb{Q}$, respectively. The duality pairing between fluid moment vectors $\mathcal{M}_n^{m_0}$ and fluid comoment vectors $\mathcal{C}_n^{m_0}$ is given by
\begin{align*}
\langle \bm{M},\bm{C}\rangle = \sum_{k=n}^{m_0} \int M_{i_1\dots i_k}(q)\,C^{i_1\dots i_k}(q)\,dq.
\end{align*}
By way of this pairing, the fluid moment vectors arise as the dual space to fluid comoment vectors, $\mathcal{M}_n^{m_0} = (\mathcal{C}_n^{m_0})^*$. 

Fluid comoment vectors with degree between $1$ and $m_0\geq 1$ comprise a Lie algebra. This result is striking because the comoment vectors with degree between $0$ and $m_0$ do not comprise a Lie algebra! The following Theorem exhibits the relevant Lie bracket and proves that it satisfies the Jacobi identity.
\begin{theorem}\label{moment_LA_thm}
Let $m_0\geq 1$ be an integer. The space $\mathcal{C}_1^{m_0}$ of fluid comoment vectors with degree between $1$ and $m_0$ is a Lie algebra. The Lie bracket between pairs of comoment vectors $\bm{C}_1,\bm{C}_2$ is given by
\begin{align}
[\bm{C}_1,\bm{C}_2]_{m_0} &= (B^1,B^2,\dots,B^{m_0})\nonumber\\
B^k &= \sum_{k_1+k_2-1 = k} \bigg(k_2\,\partial_{q^I}C_1^{(i_1\dots i_{k_1}}C_2^{i_{k_1 + 1}\dots i_{k})I}-k_1\,\partial_{q^I}C_2^{(i_1\dots i_{k_2}}C_1^{i_{k_2 + 1}\dots i_{k})I}\bigg)\,e_{i_1}\otimes \dots\otimes e_{i_k},\label{s_bracket}
\end{align}
where the sum is over all integer pairs $(k_1,k_2)$ with $k_1,k_2\geq 1$ and $k_1+k_2-1 = k$.
\end{theorem}
\begin{remark}
This formula for the Lie bracket on $\mathcal{C}_1^{m_0}$ reproduces the indicial form of the so-called Schouten concomitant\cite{Bloore_1979,Tronci_geom_2008} between symmetric contravariant tensors, modulo degree $m_0+1$. This observation is consistent with the method of proof given below since the Schouten concomitant is the unique Lie algebra structure on symmetric tensors that makes the map sending symmetric tensors to momentum polynomials a Lie algebra homomorphism. However, note that the symmetric contravariant tensors with degree between $1$ and $m_0$ are \emph{not} closed under the Schouten concomitant. Instead, these tensors possess a quotient Lie algebra structure. A similar quotient algebra structure was first noticed by Channell\cite{Channell_1995} in the context of phase space moments. Partially motivated by Channell's observations, Holm and Tronci\cite{Holm_Tronci_2009} found the Lie-Poisson bracket associated with \eqref{s_bracket} without identifying its corresponding Lie-algebraic origin. Therefore, Theorem \ref{moment_LA_thm} may be formulated in the following alternative manner: the Poisson bracket for fluid moments with degree between $1$ and $m_0$ presented in Holm and Tronci's Section 2.3 is a Lie Poisson bracket for the Lie algebra whose bracket is given by \eqref{moment_LA_thm}.
\end{remark}
\begin{proof}
First observe that for any non-negative integers $n_1\leq n_2$ there is a linear isomorphism $I_{n_1}^{n_2}:\mathcal{C}_{n_1}^{n_2}\rightarrow C_{n_1}^{n_2}(P)$ between the space $\mathcal{C}_{n_1}^{n_2}$ and the space $C_{n_1}^{n_2}(P)$ of momentum polynomials on phase space $P$ with degree between $n_1$ and $n_2$. If $\bm{C}\in \mathcal{C}_{n_1}^{n_2}$ then the corresponding momentum polynomial $I_{n_1}^{n_2}(\bm{C}) = c\in C_{n_1}^{n_2}(P)$ is
\begin{align*}
c(q,p) = \sum_{k=n_1}^{n_2} p_{i_1}\,\dots\,p_{i_k}C^{i_1\,\dots\,i_k}(q).
\end{align*}
Note the the sum over comoment degree is denoted explicitly, while the sums over tensor indices are implicit. The proof now proceeds by examining momentum polynomials in place of comoment vectors.

Let $C_{\mathfrak{s}}(P)$ denote the smooth phase space functions that vanish when $p=0$. Note that $C_1^\ell(P)\subset C_{\mathfrak{s}}(P)$ for each integer $\ell\geq 1$. The space $C_{\mathfrak{s}}(P)$ a Lie subalgebra of the Lie algebra of phase space functions under Poisson bracket. This follows immediately from the discussion of Lie subgroups of contactomorphisms in Section \ref{sec:lie}. Alternatively it can be verified directly as follows. Suppose $s_1,s_2\in C_{\mathfrak{s}}(P)$. The Poisson bracket of these functions is
\begin{align*}
\{s_1,s_2\} = \partial_{q^i}s_1\,\partial_{p_i}s_2 - \partial_{q^i}s_2\,\partial_{p_i}s_1.
\end{align*}
The sums on the right-hand-side each vanish when $p=0$ because $(\partial_{q^i}s_\ell)(q,0) = \partial_{q^i}(s_\ell(q,0)) = \partial_{q^i}(0) = 0$, $\ell = 1,2$. Therefore $\{s_1,s_2\}\in C_{\mathfrak{s}}(P)$, as claimed.

For each integer $\ell\geq 0$ let $\mathcal{I}_\ell\subset C_{\mathfrak{s}}(P)$ denote the phase space functions that vanish together with their first $\ell$ momentum derivatives when $p=0$. If $r\in \mathcal{I}_\ell$ and $s\in C_{\mathfrak{s}}(P)$ the Poisson bracket $\{r,s\}$ is
\begin{align*}
\{r,s\} = \partial_{q^i}r\,\partial_{p_i}s - \partial_{q^i}s\,\partial_{p_i}r.
\end{align*}
Differentiating the first sum with respect to momentum $\ell$ times or less produces a function that vanishes when $p=0$ because $\partial_{q_i}r$ vanishes together with its first $\ell$ momentum derivatives when $p=0$. Differentiating the second sum with respect to momentum fewer than $\ell$ times likewise produces a function that vanishes when $p=0$ because $\partial_{p_i}r$ vanishes when $p=0$ together with its first $\ell-1$ momentum derivatives. The only contribution to the $\ell^{\text{th}}$-order momentum derivative of the second sum that does not vanish at $p=0$ due to vanishing of $\partial_{p_i}r$ is a product of $\ell$ momentum derivatives of $\partial_{p_i}r$ with $\partial_{q^i}s$; $\ell^{\text{th}}$-order derivatives of $\partial_{p_i}r$ need not vanish when $p=0$. But this product vanishes when $p=0$ because $\partial_{q^i}s$ vanishes when $p=0$. Therefore $\{r,s\}\in\mathcal{I}_{\ell}$. In other words $\mathcal{I}_{\ell}$ is an ideal inside the Lie algebra $C_{\mathfrak{s}}(P)$ for each integer $\ell\geq 0$.

The quotient space $C_{\mathfrak{s}_{m_0}}(P) = C_{\mathfrak{s}}(P)/\mathcal{I}_{m_0} $ inherits the structure of a Lie algebra because $\mathcal{I}_{m_0}$ is an ideal. The Lie bracket is given by Poisson bracket mod $\mathcal{I}_{m_0}$. That is, if $s_1+\mathcal{I}_{m_0},s_2 + \mathcal{I}_{m_0}\in C_{\mathfrak{s}_{m_0}}(P)$ then the Lie bracket is
\begin{align*}
[s_1+\mathcal{I}_{m_0},s_2 + \mathcal{I}_{m_0}] = \{s_1,s_2\} + \mathcal{I}_{m_0}.
\end{align*}

There is a canonical isomorphism that sends each equivalence class in $C_{\mathfrak{s}_m}(P)$ to its unique representative that is a momentum space polynomial with degree between $1$ and $m_0$. Equivalently, $C_{\mathfrak{s}_{m_0}}(P)\approx C_1^{m_0}(P)$. Therefore $C_{1}^{m_0}(P)$ is a Lie algebra with Lie bracket given by
\begin{align*}
[c_1,c_2] = \{c_1,c_2\} \text{ mod }p^{m+1},\quad c_1,c_2\in C_1^{m_0}(P).
\end{align*}
Here $\text{ mod }p^{m+1}$ denotes setting all polynomial coefficients with degree greater than $m_0$ equal to $0$. In particular, if $c_1\in C_1^{m_0}(P)$ is homogeneous of degree $k_1\leq m_0$ and $c_2\in C_1^{m_0}(P)$ is homogeneous of degree $k_2\leq m_0$,
\begin{align*}
c_1(q,p) = p_{i_1}\dots p_{i_{k_1}}C_1^{i_1\dots i_{k_1}},\quad c_2(q,p) = p_{i_1}\dots p_{i_{k_2}}C_2^{i_1\dots i_{k_2}},
\end{align*}
the Lie bracket $[c_1,c_2]$ is given by
\begin{align}
[c_1,c_2] = \begin{cases} 0\text{ if }k_1 + k_2 - 1>m_0 \\ \bigg(k_2\,\partial_{q^I}C_1^{(i_1\dots i_{k_1}}C_2^{i_{k_1 + 1}\dots i_{k_1+k_2-1})I}-k_1\,\partial_{q^I}C_2^{(i_1\dots i_{k_2}}C_1^{i_{k_2 + 1}\dots i_{k_1+k_2-1})I}\bigg)p_{i_1}\dots p_{i_{k_1+k_2-1}} \text{ otherwise.}\end{cases}\label{monomial_bracket}
\end{align}
Here round braces around a grouping of indices denotes averaging over all permutations of that grouping.

Finally, since $\mathcal{C}_{1}^{m_0}$ is isomorphic to $C_1^{m_0}(P)$, there is a unique Lie algebra structure on comoment vectors $\mathcal{C}_1^{m_0}$ that makes $I_1^{m_0}$ an isomorphism of Lie algebras. In symbols, the bracket is $[\bm{C}_1,\bm{C}_2]_{m_0} = (I_1^{m_0})^{-1}[I_1^{m_0}(\bm{C}_1),I_1^{m_0}(\bm{C}_2)]$. Directly computing the right-hand-side of this formula with the aid of Eq.\,\eqref{monomial_bracket} reproduces Eq.\,\eqref{s_bracket}, as claimed.

\end{proof}

% 2) characterize \Gamma_{I_k}, including its image and preimages
The Lie algebra structure on the comoment vectors $\mathcal{C}_1^{m_0}$ with degree between $1$ and $m_0$ implies that the dual space $(\mathcal{C}_1^{m_0})^*$ is a Poisson manifold with a Lie-Poisson bracket. But this dual space coincides with the space of moment vectors with degree between $1$ and $m_0$, $(\mathcal{C}_1^{m_0})^* = \mathcal{M}_1^{m_0}$. The corresponding Lie-Poisson bracket $\{\cdot,\cdot\}_1^{m_0}$ on $\mathcal{M}_1^{m_0}$ comes very close to providing the sought-after Poisson structure on $P_{m_0}$. However, the degree-$0$ moments, those that correspond to density, are conspicuously missing from $\mathcal{M}_1^{m_0}$. The scalar field $\psi\in C(Q)$ that specifies the distribution center is also not contained in $\mathcal{M}_1^{m_0}$. This shortcoming of the bracket on $\mathcal{M}_1^{m_0}$, as well as the analogous bracket on phase space moments with degree between $2$ and $m_0$, has been recognized previously\cite{Channell_1995,Holm_Tronci_2009}.

Further progress comes upon noticing that the space of fluid densities $n\in C^*(Q)$ is dual to the space of scalars $\psi\in C(Q)$. This means there is a natural Poisson bracket $\{\cdot,\cdot\}_0$ on the space $C(Q)\times C^*(Q)$ given on pairs of functionals $F,G$ by
\begin{align*}
\{F,G\}_0 = \left\langle\frac{\delta F}{\delta \psi},\frac{\delta G}{\delta n} \right\rangle-\left\langle\frac{\delta G}{\delta \psi},\frac{\delta F}{\delta n} \right\rangle.
\end{align*}
It follows that there is a natural Poisson bracket on each of the factors in $P_{m_0} = \bigg(C(Q)\times C^*(Q)\bigg)\times \mathcal{M}_1^{m_0}$. The simplest Poisson bracket on $P_{m_0}$ that incorporates the Lie Poisson bracket on $\mathcal{M}_1^{m_0}$ is therefore the direct sum
\begin{align}
\{F,G\}_{m_0} &= \{F,G\}_0 + \{F,G\}_{1}^{m_0}\nonumber\\
& = \left\langle\frac{\delta F}{\delta \psi},\frac{\delta G}{\delta n} \right\rangle-\left\langle\frac{\delta G}{\delta \psi},\frac{\delta F}{\delta n} \right\rangle  + \left\langle \bm{\mu},\left[\frac{\delta F}{\delta\bm{\mu}},\frac{\delta G}{\delta \bm{\mu}}\right] \right\rangle. \label{fundamental_bracket}
\end{align}
The Jacobi identity for this bracket is automatic since the direct sum of Poisson manifolds is again a Poisson manifold. Equation \eqref{fundamental_bracket} provides the Poisson bracket for the Hamiltonian fluid moment closures presented in this Article. Its relevance to moment dynamics will be established in Section \ref{sec:exact_reduction}, where it will become clear that this bracket naturally addresses the shortcoming of the Lie-Poisson bracket on $\mathcal{M}_1^{m_0}$ in isolation. Note that, for each $m_0$, this bracket accommodates all fluid moments with degree less than $m_0$, \emph{including the density moment}. Also note that there is no restriction on the dimension of configuration space $N$.

When $m_0=0$, the bracket $\{\cdot,\cdot\}_{m_0}$ is given by
\begin{align}
\{F,G\}_0 &= \left\langle\frac{\delta F}{\delta \psi},\frac{\delta G}{\delta n} \right\rangle-\left\langle\frac{\delta G}{\delta \psi},\frac{\delta F}{\delta n} \right\rangle,\label{bracket:m0}
\end{align}
which provides the appropriate Poisson bracket for Hamiltonian fluid closures involving $\psi$ and $n$ only. This is also the bracket for the Madelung transform of the Schrodinger equation. 

When $m_0=1$, with $\mu^1= \bm{P}_0$, the bracket $\{\cdot,\cdot\}_{m_0}$ is given by
\begin{align}
\{F,G\}_1 & = \left\langle\frac{\delta F}{\delta \psi},\frac{\delta G}{\delta n} \right\rangle-\left\langle\frac{\delta G}{\delta \psi},\frac{\delta F}{\delta n} \right\rangle-\int  \bm{P}_0\cdot\bigg(\frac{\delta F}{\delta \bm{P}_0}\cdot \nabla\frac{\delta G}{\delta \bm{P}_0} -\frac{\delta G}{\delta \bm{P}_0}\cdot \nabla\frac{\delta F}{\delta \bm{P}_0}  \bigg)\,dq,\label{bracket:m1}
\end{align}
which provides the appropriate Poisson bracket for Hamiltonian fluid closures involving $\psi$ and fluid moments with degree $k<2$.

When $m_0=2$, with $\mu^2 = \mathbb{S}_0$, the bracket $\{\cdot,\cdot\}_{m_0}$ is given by
\begin{align}
\{F,G\}_2 & = \left\langle\frac{\delta F}{\delta \psi},\frac{\delta G}{\delta n} \right\rangle-\left\langle\frac{\delta G}{\delta \psi},\frac{\delta F}{\delta n} \right\rangle-\int  \bm{P}_0\cdot\bigg(\frac{\delta F}{\delta \bm{P}_0}\cdot \nabla\frac{\delta G}{\delta \bm{P}_0} -\frac{\delta G}{\delta \bm{P}_0}\cdot \nabla\frac{\delta F}{\delta \bm{P}_0}  \bigg)\,dq\nonumber\\
& -\int \mathbb{S}_0:\bigg( \frac{\delta F}{\delta \bm{P}_0}\cdot\nabla\frac{\delta G}{\delta \mathbb{S}_0} -\frac{\delta G}{\delta \bm{P}_0}\cdot\nabla\frac{\delta F}{\delta \mathbb{S}_0} \bigg)\, dq - \int \mathbb{S}_0:\bigg(\frac{\delta F}{\delta \mathbb{S}_0}\cdot \nabla\frac{\delta G}{\delta \bm{P}_0} -\frac{\delta G}{\delta \mathbb{S}_0}\cdot \nabla\frac{\delta F}{\delta \bm{P}_0}  \bigg)_s\,dq,\label{m2_bracket}
\end{align}
where $A_s = A + A^T$ for a dyad $A$.

% 3) resort to \ref{forget_thm}

\section{Hamiltonian dynamics of localized momentum distributions\label{sec:exact_reduction}}
Extending the observations of Channell\cite{Channell_1995}, Scovel-Weinstein\cite{Scovel_Weinstein_1994} provide a technical device that can be used to relate the product Poisson bracket on $P_{m_0} = C(Q)\times C^*(Q)\times\mathcal{M}_1^{m_0}$ described in Section \ref{sec:PB} to the Lie-Poisson bracket on $C^*(P)$ that underlies the Vlasov-Poisson system. The following Theorem represents an application of that device, which was originally developed for treating phase space moments, to fluid moments.

%The following Theorem is a minor corollary of Theorem 2.2 from [CITE SCOVEL-WEINSTEIN], together with the discussion at the start of that Reference's Section 4.
\begin{theorem}\label{thm:sw_map_for_moments}
Let $\{\cdot,\cdot\}_0$ denote the canonical Poisson bracket on the space $C(Q)\times C^*(Q)$. Let $\{\cdot,\cdot\}_1^{m_0}$ denote the Lie-Poisson bracket on $\mathcal{M}_{1}^{m_0}$ provided by Theorem \ref{moment_LA_thm}. Equip the space $\bigg(C(Q)\times C^*(Q)\bigg)\times \mathcal{M}_{1}^{m_0}$ with the product Poisson bracket $\{\cdot,\cdot\}_{m} = \{\cdot,\cdot\}_0 + \{\cdot,\cdot\}_1^{m_0}$.
The mapping $\Gamma_{m_0}:C(Q)\times C^*(Q)\times \mathcal{M}_{1}^{m_0}\rightarrow C^*(P)$ given by
\begin{align}
\Gamma_{m_0}(\psi,n,\bm{\mu})(q,p) &=n(q)\,\delta(p+d\psi(q)) - \mu_{i_1}(q)\,\partial_{p_{i_1}}\delta(p+d\psi(q))\nonumber\\
&+ \frac{1}{2}\,\mu_{i_1\,i_2}(q)\,\partial^2_{p_{i_1\,i_2}}\delta(p+d\psi(q)) + \dots + \frac{(-1)^{m_0}}{m!}\mu_{i_1\dots i_{m_0}}(q)\,\partial^{m_0}_{p_{i_1}\dots p_{i_{m_0}}}\delta(p+d\psi(q)).   \label{sw_poisson_map}
\end{align}
is a Poisson map.
\end{theorem}

\begin{remark}
The right-hand-side of Eq.\,\eqref{sw_poisson_map} recovers the ansatz \eqref{f_ansatz_basic} for the novel exact reduction of Vlasov-Poisson dynamics mentioned in the introduction. The fact that the map \eqref{sw_poisson_map} is Poisson justifies the claim from Section \ref{sec:PB} that the bracket $\{\cdot,\cdot\}_{m_0}$ repairs the shortcoming of the Hamiltonian fluid moment closures due to Holm and Tronci\cite{Holm_Tronci_2009} that do not include the density moment. It is interesting to note that Chong\cite{Chong_2022} previously constructed a Poisson map that sends phase space distributions into fluid-like quantities; that Poisson map appears to be unrelated to this one.
\end{remark}

\begin{proof}
Let $G$ be a Lie group with Lie subgroups $B,S\subset G$ such that the multiplication map $B\times S\rightarrow G:(b,s)\mapsto bs$ is a diffeomorphism. As summarized in Section \ref{sec:scovel}, Scovel-Weinstein proved that
\begin{align*}
\Gamma(b,b_0^*,s_0^*)= \text{Ad}^*_b(b_0^* + s_0^*),
\end{align*}
defines a Poisson map from the product Poisson manifold $(B\times \mathfrak{b}^*)\times (\mathfrak{s}^*)$ to the Lie-Poisson space $\mathfrak{g}^*$. Here $\text{Ad}^*_g$ denotes the (left) coadjoint action of $G$ on $\mathfrak{g}^*$. Let $G$ be the contactomorphism group, $B\subset G$ the fiber translations, and $S\subset G$ the isostatic contactomorphisms, as defined in Section \ref{sec:lie}. Identify the fiber translations $B$ with the Abelian group of scalar functions $\psi\in C(Q)$. Identify the Lie algebras $\mathfrak{b},\mathfrak{s},\mathfrak{g}$, with configuration space scalars $\sigma\in C(Q)\approx \mathfrak{b}$, phase space functions $s\in C_{\mathfrak{s}}(P)\approx \mathfrak{s}$ that vanish when $p=0$, and general phase space functions $h\in C(P)\approx \mathfrak{g}$. The dual spaces are then (distributional) configuration space functions $n\in C^*(Q)\approx \mathfrak{b}^*$, (distributional) phase space functions $\tilde{f}_0\in C^*_{\mathfrak{s}}(P)$ with zero fluid density, and general (distributional) phase space functions $f\in C^*(P)\approx \mathfrak{g}^*$. The corresponding duality pairings are the usual $L^2$ integration pairings.  Under these identifications, the Poisson map $\Gamma: C(Q)\times C^*(Q)\times C^*_{\mathfrak{s}}(P)\rightarrow C^*(P)$ becomes
\begin{align*}
\Gamma(\psi,n,\tilde{f}_0)(q,p) &= \tau_{\psi*}(\pi^*n\,\delta(p) + \tilde{f}_0)(q,p)= n(q)\,\delta(p + d\psi(q)) + \tilde{f}_0(q,p+d\psi(q)).
\end{align*}
Note that $C(Q)\times C^*(Q)\times C^*_{\mathfrak{s}}(P)$ is endowed with the product Poisson structure corresponding to the canonical poisson bracket on $C(Q)\times C^*(Q)$ and the Lie Poisson bracket on $C^*_{\mathfrak{s}}(P)$.

As shown in the proof of Theorem \ref{moment_LA_thm}, the phase space functions $r\in \mathcal{I}_{m_0}$ that vanish along with their first $m_0$ momentum derivatives when $p=0$ form an ideal in the Lie algebra $C_{\mathfrak{s}}(P)$. Therefore the quotient space $C_{\mathfrak{s}}(P)/\mathcal{I}_{m_0} = C_{\mathfrak{s}_{m_0}}(P)$ has a natural Lie algebra structure and the quotient map $\Pi_{m_0}:C_{\mathfrak{s}}(P) \rightarrow  C_{\mathfrak{s}_{m_0}}(P)$ is a Lie algebra homomorphism. Also from the proof of Theorem \ref{moment_LA_thm}, the quotient space $C_{\mathfrak{s}_{m_0}}(P)$ is naturally isomorphic as a Lie algebra to the space of comoment vectors $\mathcal{C}_1^{m_0}$ with degree between $1$ and $m_0$. (The Lie algebra structure on $\mathcal{C}_1^{m_0}$ is the one provided by Theorem \ref{moment_LA_thm}.) Under this identification, $\Pi_{m_0}:C_{\mathfrak{s}}(P)\rightarrow \mathcal{C}_1^{m_0}$ is the map that assigns to each $s\in C_{\mathfrak{s}}(P)$ its weighted momentum Taylor series coefficients with degrees between $1$ and $m_0$:
\begin{align*}
\Pi_{m_0}(s) = \bigg(\partial_{p_{i_1}}s(q,0)\,e_{i_1},\frac{1}{2}\partial^2_{p_{i_1}\,p_{i_2}}s(q,0)\,e_{i_1}\otimes e_{i_2},\dots, \frac{1}{m!}\partial^{m_0}_{p_{i_1}\dots p_{i_{m_0}}}s(q,0)\,e_{i_1}\otimes\dots \otimes e_{i_{m_0}}\bigg).
\end{align*}
Since $\Pi_{m_0}$ is a Lie algebra homomorphism, the adjoint $\Pi_{m_0}^*:\mathcal{M}_1^{m_0}\rightarrow C^*_{\mathfrak{s}}(P)$ is a Poisson mapping between Lie-Poisson spaces. A simple direct calculation leads to the explicit formula
\begin{align*}
\Pi_{m_0}^*(\bm{\mu})(q,p) = - \mu_{i_1}(q)\,\partial_{p_{i_1}}\delta(p) + \frac{1}{2}\,\mu_{i_1\,i_2}(q)\,\partial^2_{p_{i_1\,i_2}}\delta(p) + \dots + \frac{(-1)^{m_0}}{m_0!}\mu_{i_1\dots i_{m_0}}(q)\,\partial^{m_0}_{p_{i_1}\dots p_{i_{m_0}}}\delta(p).
\end{align*}

Since $\Pi_{m_0}^*$ and $\Gamma$ are Poisson maps, the composition
\begin{align*}
\Gamma_{m_0}:C(Q)\times C^*(Q)\times \mathcal{M}_1^{m_0}\xrightarrow{\text{id}_{C(Q)\times C^*(Q)}\times \Pi_{m_0}^*} C(Q)\times C^*(Q)\times C^*_{\mathfrak{s}}(P)\xrightarrow{\Gamma} C^*(P),
\end{align*}
is a Poisson map as well. Evaluating the composition using the above formulas shows that $\Gamma_{m_0}$ is the map given in the Theorem statement.
\end{proof}

% 1) Show how the physical Hamiltonian for VP on \mathfrak{g}^* induces Hamiltonians on SW, M
Theorem \ref{thm:sw_map_for_moments} is quite powerful due to a construction\cite{Guillemin_1980} known as ``collectivisation." Given a Poisson mapping $\Gamma:\Lambda\rightarrow Z$ between Poisson manifolds $\Lambda,Z$, every Hamiltonian function $h$ on $Z$ gives rise to a corresponding Hamiltonian function $H =\Gamma^*h$ on $\Lambda$ known as the ``collective" Hamiltonian. Remarkably, if $\lambda(t)$ is any integral curve for the collective Hamiltonian vector field $X_H$ then the image curve $z(t)=\Gamma(\lambda(t))$ is an integral curve for the Hamiltonian vector field $X_h$. This follows from the simple chain of identities for arbitrary $q\in C(Z)$:
\begin{align*}
\frac{d}{dt}q(z(t)) = \frac{d}{dt}q(\Gamma(\lambda(t))) = \frac{d}{dt}(\Gamma^*q)(\lambda(t)) = \{\Gamma^*q,\Gamma^*h\}_{\Lambda}(\lambda(t)) = (\Gamma^*\{q,h\}_Z)(\lambda(t)) = \{q,h\}_Z(z(t)).
\end{align*}

In the context of fluid moments, the collectivisation construction and Theorem \ref{thm:sw_map_for_moments} have the following implication. Let 
\begin{align}
\mathcal{H}(f) = \int\frac{1}{2m}|p|^2f\,dq\,dp + \int\frac{1}{2}e\,\hat{V}(n)\,n\,dq 
\end{align}
denote the Hamiltonian for the Vlasov-Poisson system. Here $\hat{V}:C^*(Q)\rightarrow C(Q)$ is the self-adjoint linear operator that sends a density $n$ to the corresponding solution $V$ of Poisson's equation, $-\Delta V = 4\pi e\,\widetilde{n} $, where $\widetilde{n}$ denotes the number density less its average. Pulling back $\mathcal{H}$ along the Poisson map $\Gamma_{m_0}$ produces a collective Hamiltonian $\mathcal{H}_{m_0}$ on $P_{m_0} = C(Q)\times C^*(Q)\times\mathcal{M}_1^{m_0}$. If $(\psi(t),n(t),\bm{\mu}(t))$ is any solution of the Hamiltonian system on $P_{m_0}$ defined by $\mathcal{H}_{m_0}$ then $f(t) = \Gamma_{m_0}(\psi(t),n(t),\bm{\mu}(t))$ is an exact solution of the Vlasov-Poisson system. In other words, Theorem \ref{thm:sw_map_for_moments} produces exact solutions of Vlasov-Poisson with any finite number $m_0$ of fluid moments!

When $m_0=0$, the space $P_{m_0}$ is just $P_0 = C(Q)\times C^*(Q)\ni (\psi,n)$, with the standard canonical Poisson bracket \eqref{bracket:m0}. Pulling back the physical Hamiltonian $\mathcal{H}(f)$ along $\Gamma_0$ leads to the collective Hamiltonian 
\begin{align*}
\mathcal{H}_{0}(\psi,n) = \frac{1}{2m}\int |\nabla\psi|^2\,n\,dq + \frac{1}{2}\int e\,\hat{V}(n)\,n\,dq.
\end{align*}
The functional derivatives of the collective Hamiltonian are given by
\begin{align*}
\frac{\delta \mathcal{H}_{0}}{\delta \psi} = -m^{-1}\nabla\cdot(n\,\nabla\psi),\quad \frac{\delta \mathcal{H}_{0}}{\delta n} = \frac{1}{2m}|\nabla\psi|^2 + e \hat{V}(n).
\end{align*}
The corresponding equations of motion are therefore
\begin{align*}
\partial_tn & = m^{-1}\nabla\cdot (n\,\nabla\psi)\\
\partial_t\psi &= \frac{1}{2m}|\nabla\psi|^2  +e\hat{V}(n).
\end{align*}
Upon introducing the velocity variable $\bm{v} = -m^{-1}\nabla\psi$, these equations may be re-written as
\begin{align*}
\partial_t(m\,n\,\bm{v}) + \nabla\cdot (m\,n\,\bm{v}\,\bm{v}) = -e\,n\,\nabla \hat{V}(n),\quad \partial_tn + \nabla\cdot(n\,\bm{v})=0,\quad m\,\bm{v} = -\nabla\psi,
\end{align*}
which reveals equivalence with Euler's equations for a cold, irrotational fluid. It is a well-known fact\cite{Uhlemann_2018} that such irrotational flows correspond to exact solutions of the Vlasov-Poisson system. 

When $m_0=1$, with $\mu^1 = \bm{P}_0$, the space $P_{m_0}$ becomes $P_1 = C(Q)\times C^*(Q)\times\mathcal{M}_1^1$ equipped with the Poisson bracket in Eq.\,\eqref{bracket:m1}. Pulling back the physical Hamiltonian along $\Gamma_1$ leads to the collective Hamiltonian
\begin{align*}
\mathcal{H}_1(\psi,n,\bm{P}_0) =\int \frac{1}{2m n}|\bm{P}_0 - n\nabla\psi|^2\,dq -\int \frac{1}{2mn}|\bm{P}_0|^2\,dq + \frac{1}{2}\int e\,\hat{V}(n)\,n\,dq. 
\end{align*}
The functional derivatives of the collective Hamiltonian are
\begin{align*}
\frac{\delta\mathcal{H}_1}{\delta\psi} = m^{-1}\nabla\cdot(\bm{P}_0 - n\nabla\psi),\quad \frac{\delta \mathcal{H}_1}{\delta n} =\frac{1}{2m}|\nabla\psi|^2 + e\hat{V}(n),\quad \frac{\delta\mathcal{H}_1}{\delta \bm{P}_0} = -m^{-1}\nabla\psi.
\end{align*}
The corresponding equations of motion are therefore
\begin{align*}
\partial_tn &= -m^{-1}\nabla\cdot(\bm{P}_0 - n\,\nabla\psi)\\
\partial_t\psi & = \frac{1}{2m}|\nabla\psi|^2 + e\hat{V}(n)\\
\partial_t\bm{P}_0 & = m^{-1}\bm{P}_0\cdot\nabla(\nabla\psi) + m^{-1}(\nabla\psi)\cdot\nabla \bm{P}_0 + m^{-1}\nabla\cdot(\nabla\psi)\,\bm{P}_0.
\end{align*}
Distributions in the image of $\Gamma_1$ have momentum density $\bm{P} = \int p\,f\,dp = \bm{P}_0 - n\,\nabla\psi$ and momentum flux 
\begin{align*}
\mathbb{S} = m^{-1}\int p\,p\,f\,dp= -m^{-1}\,\nabla\psi\,\bm{P}_0 - m^{-1}\,\bm{P}_0\nabla\psi + m^{-1}\,n\,\nabla\psi\nabla\psi= -m^{-1}\,\nabla\psi\,\bm{P} - m^{-1}\,\bm{P}\nabla\psi - m^{-1}\,n\,\nabla\psi\nabla\psi.
\end{align*}
In terms of these variables, the equations of motion are instead
\begin{align*}
\partial_t(mn) +\nabla\cdot\bm{P}=0,\quad \partial_t\bm{P} + \nabla\cdot \mathbb{S} = -e\,n\,\nabla \hat{V}(n),\quad  \partial_t\psi & = \frac{1}{2m}|\nabla\psi|^2 + e\hat{V}(n).
\end{align*}
These equations reveal that the exact reduction with $m_0=1$ again recovers familiar hydrodynamic equations, but now with a non-trivial closure relation for the momentum flux tensor $\mathbb{S}$. This Hamiltonian moment closure appears to be new, but it would be interesting to compare it with previously-derived Hamiltonian moment closures that include moments with degree $k<2$.

When $m_0=2$, with $\mu^2 = \mathbb{S}_0$, the space $P_{m_0} = C(Q)\times C^*(Q)\times\mathcal{M}_1^2$ is equipped with the Poisson bracket \eqref{m2_bracket}. The collective Hamiltonian is
\begin{align*}
\mathcal{H}_2(\psi,n,\bm{P}_0,\mathbb{S}_0) & = \frac{1}{2m}\int \text{Tr}(\mathbb{S}_0)\,dq -\frac{1}{m}\int \bm{P}_0\cdot \nabla\psi\,dq + \frac{1}{2m}\int |\nabla\psi|^2\,n + \frac{1}{2}\int e\,\hat{V}(n)\,n\,dq.
\end{align*}
The corresponding functional derivatives are
\begin{align*}
\frac{\delta\mathcal{H}_2}{\delta \psi} = \frac{1}{m}\nabla\cdot(\bm{P}_0 - n\nabla\psi),\quad \frac{\delta\mathcal{H}_2}{\delta n}  = \frac{1}{2m}|\nabla\psi|^2 + e\,\hat{V}(n),\quad \frac{\delta \mathcal{H}_2}{\delta \bm{P}_0} =-\frac{1}{m}\nabla\psi,\quad \frac{\delta\mathcal{H}_2}{\delta \mathbb{S}_0} = \frac{1}{2m}\eta = \frac{1}{2m}\delta^{i_1\,i_2}\,e_{i_1}\otimes e_{i_2}.
\end{align*}
The equations of motion are therefore given by
\begin{align*}
\partial_t\psi & = \frac{1}{2}m|\bm{v}|^2 + e\,\hat{V},\quad \bm{v} = -m^{-1}\nabla\psi\\
\partial_tn & = -m^{-1}\nabla\cdot(\bm{P}_0 + m\,n\bm{v})\\
\partial_t\bm{P}_0 & = -\bm{v}\cdot\nabla\bm{P}_0 - \nabla\bm{v}\cdot\bm{P}_0 - (\nabla\cdot\bm{v})\bm{P}_0 - m^{-1}\nabla\cdot\mathbb{S}_0\\
\partial_t\bm{S}_0 & =-\bm{v}\cdot\nabla\mathbb{S}_0 - \nabla\bm{v}\cdot \mathbb{S}_0 - \mathbb{S}_0\cdot\nabla\bm{v} - (\nabla\cdot\bm{v})\mathbb{S}_0.
\end{align*}
Note that these equations imply $\partial_t\bm{v} + \bm{v}\cdot\nabla\bm{v} = -\frac{e}{m}\nabla\hat{V}$. Upon introducing the conventional (non-centered) moments,
\begin{align*}
\bm{P} & = \bm{P}_0 + m\,n\,\bm{v}\\
\mathbb{S} & = \mathbb{S}_0 + m\,(\bm{v}\,\bm{P}_0 + \bm{P}_0\,\bm{v}) + m^2\,n\,\bm{v}\,\bm{v},
\end{align*}
and performing a tedious explicit calculation, the evolution equations for $(\psi,n,\bm{P}_0,\mathbb{S}_0)$ may be re-written as
\begin{align}
\partial_t\psi & = \frac{1}{2}m|\bm{v}|^2 + e\,\hat{V},\quad \bm{v} = -m^{-1}\nabla\psi\label{2o_one}\\
\partial_tn & = - m^{-1}\nabla\cdot\bm{P}\\
\partial_t\bm{P} & = - m^{-1}\nabla\cdot\mathbb{S} -e\,n\,\nabla\hat{V}\\
\partial_t\mathbb{S} & = -m^{-1}\nabla\cdot\mathcal{T} - e\,(\nabla\hat{V}\bm{P} + \bm{P}\nabla\hat{V}).\label{2o_four}
\end{align}
These evolution equations for $(n,\bm{P},\mathbb{S})$ exactly recover the first three equations in the moment hierarchy \eqref{momentum_flux_eqn} with a novel closure for the third-degree fluid moment $\mathcal{T} = \int p\,p\,p\,f\,dp$ given by
\begin{align}
\mathcal{T} & = 3m\,\bm{v} \odot \mathbb{S}_0 + 3m^2\,\bm{v}\bm{v}\odot \bm{P}_0 + m^3\,n\,\bm{v}\bm{v}\bm{v}\nonumber\\
& = 3m\,\bm{v}\odot \mathbb{S} - 3m^2\,\bm{v}\bm{v}\odot\bm{P} + m^3\,n\,\bm{v}\bm{v}\bm{v}.
\end{align}
Here $\odot$ denotes the symmetrized tensor product, e.g. $3(\bm{v} \odot \mathbb{S}_0)_{i_1\,i_2\,i_3} = v_{i_1}\,S_{i_2\,i_3} + v_{i_2}\,S_{i_1\,i_3} + v_{i_3}\,S_{i_1\,i_2}$. 

The system \eqref{2o_one}-\eqref{2o_four} admits stationary homogeneous equilibria of the form $\psi=0$, $n=n_0=\text{const.}$, $\bm{P}=0$, $\mathbb{S} = \mathbb{S}_0=\text{const.}$ When $N=1$ and $\mathbb{S}_0\geq 0$, these equilibria are stable if and only if $\mathbb{S}_0 = 0$. When $\mathbb{S}_0 > 0$ the linearized spectrum for Fourier mode $k$ contains a pair of imaginary eigenvalues corresponding to oscillations at the Bohm-Gross frequency,
\begin{align*}
\omega^2 = \omega_p^2\bigg(\frac{1}{2} + \frac{1}{2}\sqrt{1 + 4 v_{\text{th}}^2\,k^2/\omega_p^2}\bigg),\quad \omega_p^2 = \frac{4\pi\,e^2n_0}{m},\quad v_{\text{th}}^2 = \frac{\mathbb{S}_0}{m^2\,n_0},
\end{align*}
and a pair of real eigenvalues corresponding to exponential growth and decay with rates
\begin{align*}
\gamma^2 = -\omega_p^2\bigg(\frac{1}{2} - \frac{1}{2}\sqrt{1 + 4 v_{\text{th}}^2\,k^2/\omega_p^2}\bigg)\bigg).
\end{align*}
(Curiously, equilibria with $\mathbb{S}_0 < 0$ are linearly stable!) So the exact closure with $m_0=2$ already begins to reflect the well-known fact that only certain homogeneous equilibria for the Vlasov-Poisson system, e.g. monotone-decreasing distributions\cite{MP_1989}, are linearly stable. Note that the condition $\mathbb{S}_0 > 0$ does not imply a stable distribution; there are many linearly-unstable homogeneous Vlasov-Poisson equilibria with positive second moment. Also note that the growth rate $\gamma=O(\sqrt{|k|})$ for large $k$, implying ill-posedness of the initial value problem. Physical solutions must therefore be found within the center manifolds\cite{Fenichel_1979,MacKay_2004,Burby_Klotz_2020} attached to such equilibria, much as physical solutions of the Abraham-Lorentz-Dirac equation for individual electrons are contained in a center manifold\cite{Spohn_2000}. Further work is needed to determine linear stability properties of homogeneous equilibria with $N>1$, $m_0>2$. It is unclear whether the data-driven Hamiltonians described in Section \ref{sec:general_scheme} will also require reduction to a center manifold to describe physical dynamics.

% The distributions in the image of the Scovel-Weinstein Poisson map are never positive linear functionals unless $\mu_i=0$ for $i\geq 1$. This presents a theoretical conundrum that can be understood as follows. The space $\mathcal{M}_1^{m_0}$ has been constructed as the dual to a space of momentum polynomials, not the entire space of phase space functions. Therefore it is sensible to replace the notion of positive linear functionals of phase space functions with the notion of positive linear functionals of momentum polynomials. For example, the positive linear functionals on momentum polynomials with degree between $1$ and $2$ are those with $\mu_2$ positive semi-definite. From this perspective, the image of the Scovel-Weinstein Poisson map contains distributions that may be understood as positive linear functionals in a weak sense. 

% 2) prove that every Hamiltonian traj in M corresponds to a Hamiltonian trajectory in \mathfrak{g}^* with momentum localization

\section{A closure scheme for general distributions\label{sec:general_scheme}}
% {\color{red}General formula for Hamiltonian vector field $m=0,1,2$}

% {\color{red}Formulas for passing between centered moments and moments}

% {\color{red}Formulate learning task for $\mathcal{H}$}

% 1) define a SW "model"

% 2) show that models lead to different moment dynamics when free energy is used in place of energy

The theory developed so far leads to a framework for constructing data-driven fluid closures of the Vlasov-Poisson system with Hamiltonian structure. This framework assumes availability of a legacy simulation code for the Vlasov-Poisson system; either particle-in-cell (PIC) codes or continuum codes will suffice. The framework presented here should be contrasted with the structure-preserving reduced-order modeling techniques developed by Hesthaven, Pagliantini, and Ripamonti\cite{Hesthaven_rank_2020,Hesthaven_poisson_2021,Hesthaven_HAM_2021}. Those reduced-order models assume the reduced, data-driven dependent variables admit constant symplectic or Poisson structures, or else rely on transformations to such variables. In contrast, the method presented here assumes the reduced dependent variables admit the non-constant Poisson structure provided by Scovel-Weinstein theory. In addition, the reduced dependent variables used here are always simply related to physical quantities of interest, namely the fluid moments.

% Suppose a legacy simulation code for the Vlasov-Poisson system assigns to each initial distribution function $f_0(x,p) = f(x,p,0)$ a sequence of distribution functions $\bm{f}(f_0) = (f_0,f_1(f_0),\dots,f_N(f_0))$,  $f_i(x,p) =f(x,p,t_i)$, that approximates the time evolution of $f_0$ on the finite temporal grid $t_i\in\{0,t_1,t_2,\dots,t_N\}$. Suppose further that for each initial distribution function $f_0$ in a finite dataset $\mathcal{F}_0$ the legacy simulation code has been run to produce the corresponding approximate time evolution sequence $\bm{f}(f_0)$. The pairs $(f_0,\bm{f}(f_0))$ then comprise a set of labeled training data.

Developing a closure that includes fluid moments up to degree-$m_0$ requires a parametric energy function
\begin{align*}
\mathcal{H}_W : C(Q)\times \mathcal{M}_0^{m_0} \rightarrow\mathbb{R}:(\psi,\bm{\mu})\mapsto \mathcal{H}_W(\psi,\bm{\mu}),
\end{align*}
and a \emph{model} 
\begin{align*}
\mathcal{F}_0:\mathcal{M}_0^{m_0}\rightarrow C^*(P)
\end{align*}
that assigns a centered distribution $f_0 = \mathcal{F}_0(\bm{\mu})$ to each centered moment vector $\bm{\mu}\in \mathcal{M}_0^{m_0}$. This terminology has been borrowed from Scovel-Weinstein\cite{Scovel_Weinstein_1994}, who introduced the notion of a model for phase space moments. For consistency, the model distribution function $\mathcal{F}_0(\bm{\mu})$ must have moments compatible with $\bm{\mu}$:
\begin{align*}
\mu^k_{i_1\,\dots\,i_k} = \int p_{i_1}\,\dots \,p_{i_k}\mathcal{F}_0(\bm{\mu})\,dq,\quad k=0,1,\dots,m_0.
\end{align*}
Given a field $\psi\in C(Q)$, the ordinary distribtion function $f$ associated with $\mathcal{F}_0(\bm{\mu})$ is $f = \tau_{\psi *}\mathcal{F}_0(\bm{\mu})$. To fully specify the closure, the variable $W$ that parameterizes the energy must be adjusted to an optimal value $W^*$. The closure itself will then comprise a Hamiltonian system on $C(Q)\times \mathcal{M}_0^{m_0}$ with Hamiltonian $\mathcal{H}_{W^*}$ and Poisson bracket given by Eq.\,\eqref{fundamental_bracket}.

Finding the optimal parameter value $W^*$ proceeds as follows. Start by selecting a finite collection $\mathcal{D}$ of samples $(\psi,\bm{\mu})\in \mathcal{D}$ from the space $C(Q)\times\mathcal{M}_0^{m_0}$, together with a timestep parameter $\Delta t\in\mathbb{R}$. For each sample $(\psi,\bm{\mu})\in \mathcal{D}$ perform two computations. (1) Approximately evolve $(\psi,\bm{\mu})$ along the Hamiltonian vector field $X_{\mathcal{H}_W} = (\dot{\psi}_{\mathcal{H}_W},\dot{\bm{\mu}}_{\mathcal{H}_W})$ for $\Delta t$ seconds using the forward-Euler formula
\begin{align*}
\psi(\Delta t) = \psi + \Delta t\,\dot{\psi}_{\mathcal{H}_W}(\psi,\bm{\mu}),\quad \bm{\mu}(\Delta t) = \bm{\mu} +\Delta t\, \dot{\bm{\mu}}_{\mathcal{H}_W}(\psi,\bm{\mu}).
\end{align*}
Then compute the corresponding candidate evolved moment vector $\overline{\bm{M}}(\Delta t) = (\overline{M}^0(\Delta t),\dots,\overline{M}^{m_0}(\Delta t))$, where the degree-$k$ candidate moment is given by
\begin{align}
\overline{M}(\Delta t)_{i_1\,\dots\,i_k} = \int p_{i_1}\dots p_{i_k}\tau_{\psi(\Delta t) *}\mathcal{F}_0(\bm{\mu}(\Delta t))\,dp.\label{candidate_moments}
\end{align}
Note that $\overline{M}^k(\Delta t)$ depends explicitly on $W$.
(2) Reconstruct the distribution function $f = \tau_{\psi *}\mathcal{F}_0(\bm{\mu})$. Use the legacy simulation code to find the time evolution $f(\Delta t)$ of $f$. Then compute the corresponding ``ground-truth" evolved moment vector $\bm{M}(\Delta t)$, where the degree-$k$ ground-truth evolved moment is given by
\begin{align}
{M}_{i_1\,\dots\,i_k} = \int p_{i_1}\dots p_{i_k}f(\Delta t)\,dp.\label{gt_moments}
\end{align}
With each sample $(\psi,\bm{\mu})\in\mathcal{D}$ now processed, construct the mean-squared loss
\begin{align*}
S(W) = \frac{1}{|\mathcal{D}|}\sum_{(\psi,\bm{\mu})\in\mathcal{D}} \| \bm{M}(\Delta t) - \overline{\bm{M}}(\Delta t) \|^2,
\end{align*}
and find $W^*$ by solving the optimization problem
\begin{align*}
W^* = \text{argmin}\,S(W).
\end{align*}
In this way, the energy parameters $W$ are chosen in order to maximize agreement between moment evolution predicted by the Vlasov-Poisson system and moment evolution predicted by the Hamiltonian hydrodynamic model for the field $\psi$ and the centered moments $\bm{\mu}$.

Through application of the binomial formula for the degree-$k$ candidate evolved moment
\begin{align}
\overline{M}^k(\Delta t) = \sum_{\ell=0}^k (-1)^\ell \begin{pmatrix} k\\ \ell\end{pmatrix} (\nabla\psi(\Delta t))^\ell\odot \mu^{k-\ell}(\Delta t),\label{binomial_formula}
\end{align}
the $W$-dependence of the loss $S(W)$ can be made more explicit. For example, when $m_0=1$, with $\mu^0 = n$,  $\mu^1 \equiv \bm{P}_0$, the general Hamiltonian vector field $X_{\mathcal{H}} = (\dot{\psi}_{\mathcal{H}},\dot{n}_{\mathcal{H}},\dot{\bm{P}}_{0\,\mathcal{H}})$ is given by
\begin{align*}
\dot{\psi}_{\mathcal{H}} & = \frac{\delta\mathcal{H}}{\delta n}\\
\dot{n}_{\mathcal{H}} & = -\frac{\delta\mathcal{H}}{\delta \psi}\\
\dot{\bm{P}}_{0\,\mathcal{H}} & = -\frac{\delta\mathcal{H}}{\delta \bm{P}_0}\cdot \nabla\bm{P}_0 - \bigg(\nabla\frac{\delta\mathcal{H}}{\delta\bm{P}_0}\bigg)\cdot \bm{P}_0 - \bigg(\nabla\cdot\frac{\delta\mathcal{H}}{\delta \bm{P}_0}\bigg)\bm{P}_0.
\end{align*}
The loss $S(W)$ with $m=1$ may therefore be written explicitly as
\begin{align*}
S(W) &= \frac{1}{|\mathcal{D}|}\sum_{(\psi,\bm{\mu})\in\mathcal{D}} \int \left|n - \Delta t\frac{\delta \mathcal{H}_{W}}{\delta\psi} - \int f(\Delta t)\,dp \right|^2\,dq\\
& +\frac{1}{|\mathcal{D}|}\sum_{(\psi,\bm{\mu})\in\mathcal{D}}\int \bigg| \bm{P}_0 - \Delta t\frac{\delta\mathcal{H}_W}{\delta\bm{P}_0}\cdot \nabla\bm{P}_0 - \Delta t\,\bigg(\nabla\frac{\delta\mathcal{H}_W}{\delta\bm{P}_0}\bigg)\cdot \bm{P}_0 - \Delta t\bigg(\nabla\cdot \frac{\delta\mathcal{H}_W}{\delta\bm{P}_0}\bigg)\bm{P}_0\\
&\hspace*{13em}- \left(n - \Delta t\frac{\delta\mathcal{H}_W}{\delta \psi}\right)\nabla\left(\psi + \Delta t\frac{\delta\mathcal{H}_W}{\delta n}\right) - \int p\,f(\Delta t)\,dp\bigg|^2\,dq.
\end{align*}

\section{Discussion}
% 1) argue that NN models can produce data-driven Hamiltonian moment closures
This Article has described a new method for developing fluid moment closures of the Vlasov-Poisson system with Hamiltonian structure. These closures apply in any number of space dimensions, and allow for collections of fluid moments with degree between $0$ and $m_0$, where $m_0$ is any non-negative integer. The closures work by introducing an irrotational velocity field $\bm{v}$ that specifies the ``center" of the distribution function in moment space. The theory presented here leads to both novel exact solutions of the Vlasov-Poisson system parameterized by arbitrarily-large collections of fluid moments, as well as a general framework for finding data-driven fluid closures with Hamiltonian structure.

% 2) how can one achieve Hamiltonian hyperbolicity?
Variable-moment data-driven fluid closures garner increasing attention in the machine learning community. Recent results\cite{Hauck_2021,Huang_Christlieb_II_2021,Huang_Christlieb_III_2023} indicate that it is possible to construct such closures while strictly enforcing hyperbolicity. The results presented in the present Article therefore suggest investigation of the following intriguing question. Is it possible to develop expressive variable-moment data-driven fluid closures that strictly enforce both hyperbolicity \emph{and} Hamiltonian structure? Within the closure framework presented here, such a constraint would limit the allowable forms of the closure Hamiltonian $\mathcal{H}_W(\psi,n,\bm{\mu})$.

% 3) The Hamiltonian grad method includes all moments up to a given order. Are there other Hamiltonian moment closures that leave some moments out?
The fluid closures presented here always include all fluid moments up to and including some maximum degree $m_0$. It would be interesting to determine whether these closures can be specialized to only include a subset of these moments. For example, are there subclosures with scalar pressure?

% 4) Say that the electromagnetic problem will be solved next.
This Article has treated the electrostatic limit of collisionless plasma kinetic theory. Variable-moment fluid closures with Hamiltonian structure would also be interesting to develop for the electromagnetic problem, as embodied by the Vlasov-Maxwell system. Forthcoming publications will reveal that the Lie-theoretic methods used here transfer to the electromagnetic problem in a mathematically and physically appealing fashion.

% 5) Casimirs, stability
Each of the fluid closures developed here is associated with a certain direct-sum Poisson bracket $\{\cdot,\cdot\}_{m_0}$. (See Eq.\,\eqref{fundamental_bracket}.) It would be useful to identify Casimir functions for this bracket and use them to study stability of solutions using the energy-Casimir method\cite{Holm_1985}. Finding the Casimirs would also be a necessary step when including dissipation in the fluid closures presented here using a metriplectic\cite{Morrison_met_1986} bracket.

As mentioned in the introduction, Tassi\cite{Tassi_2015,Tassi_2022,Tassi_2023} has found variable-moment fluid closures for slab drift-kinetic systems. But the methods used by Tassi to find closures are seemingly unrelated to the Lie-theoretic methods of Scovel-Weinstein. Is it possible to find Tassi's drift kinetic closures using Scovel-Weinstein theory\cite{Scovel_Weinstein_1994} as was done for the closures presented here? An affirmative result might suggest a path toward variable-moment gyrofluid closures for more intricate gyroaveraged kinetic theories such as full-$f$ gyrokinetics\cite{Brizard_2007,Burby_2015,Hirvijoki_2022}. In a similar vein, it would be interesting to determine whether the closures presented here can be understood as limits of the multiple waterbag reductions described by Perin \emph{et al.}\cite{PCMT_2015_waterbag}.

More broadly, this Article suggests that the methods developed by Scovel-Weinstein\cite{Scovel_Weinstein_1994} in the context of beam dynamics apply to a wider class of dissipation-free statistical closure problems than originally envisioned. There are many possible extensions. For example, consider the problem of solving the functional Vlasov equation for a variational weakly-nonlinear wave equation with random initial data. The algebra of functionals on the field-theoretic phase space enjoys a canonical Poisson bracket. By decomposing this algebra into affine functionals and functionals that vanish together with their first functional derivatives for the vacuum state, Scovel-Weinstein theory may be applied. The result will be Hamiltonian evolution equations for a finite collection of field correlation functions. These non-dissipative correlation dynamics do not seem to have been studied previously.

\section*{Data availability}
No datasets were generated or analysed during the current study.

\bibliography{sample}

% \noindent LaTeX formats citations and references automatically using the bibliography records in your .bib file, which you can edit via the project menu. Use the cite command for an inline citation, e.g.  .

% For data citations of datasets uploaded to e.g. \emph{figshare}, please use the \verb|howpublished| option in the bib entry to specify the platform and the link, as in the \verb|Hao:gidmaps:2014| example in the sample bibliography file.

% \begin{align*}
% T &= T_{i_1\,i_2}\,e_{i_1}\otimes e_{i_2}\\
% S &= S_{i_1\,i_2\,i_3}\,e_{i_1}\otimes e_{i_2} \otimes e_{i_3}\\
% A & = T\otimes S = T_{i_1\,i_2}\,S_{i_3\,i_4\,i_5}\,e_{i_1}\otimes \dots \otimes e_{i_5}\\
% w & = w_{i_1}\,e_{i_1}\\
% w\cdot T & = w_{i_1}\,T_{i_1\,i_2} \,e_{i_2} \\
% C_{23}S & = S_{i_1\,I\,I}\,e_{i_1}\\
% C_{13}S & = S_{I\,i_2\,I}\,e_{i_2}\\
%  C_{13}(w\otimes S) & = w_I S_{i_2\,I\,i_4}\,e_{i_2}\otimes e_{i_4}
% \end{align*}

\section*{Acknowledgements}

Research reported in this Article was supported by a US DOE Mathematical Multifaceted Integrated Capability Center (MMICC) grant. This Article benefits from several helpful discussions with C. Chandre, A. Christlieb, P. Morrison, E. Tassi, and C. Tronci.

\section*{Author contributions statement}
J.W.B. performed all work relevant to the preparation of this manuscript.

\section*{Additional information}
The author declares no competing interests. 

% To include, in this order: \textbf{Accession codes} (where applicable); 

% The corresponding author is responsible for submitting a \href{http://www.nature.com/srep/policies/index.html#competing}{competing interests statement} on behalf of all authors of the paper. This statement must be included in the submitted article file.

% \begin{figure}[ht]
% \centering
% \includegraphics[width=\linewidth]{stream}
% \caption{Legend (350 words max). Example legend text.}
% \label{fig:stream}
% \end{figure}

% \begin{table}[ht]
% \centering
% \begin{tabular}{|l|l|l|}
% \hline
% Condition & n & p \\
% \hline
% A & 5 & 0.1 \\
% \hline
% B & 10 & 0.01 \\
% \hline
% \end{tabular}
% \caption{\label{tab:example}Legend (350 words max). Example legend text.}
% \end{table}

% Figures and tables can be referenced in LaTeX using the ref command, e.g. Figure \ref{fig:stream} and Table \ref{tab:example}.

\end{document}